\documentclass{article}
\usepackage{ijcai21}

\usepackage{times}
\usepackage{soul}
\usepackage{url}
\usepackage[hidelinks]{hyperref}
\usepackage[utf8]{inputenc}
\usepackage[small]{caption}
\usepackage{graphicx}
\usepackage{amsmath}
\usepackage{amsthm}
\usepackage{booktabs}
\usepackage{algorithm}
\newtheorem{theorem}{Theorem}

\usepackage{amssymb,amsfonts}
\usepackage{subcaption}
\usepackage{thmtools,thm-restate}
\usepackage{color}
\usepackage[noend]{algpseudocode}
\usepackage{enumerate}
\usepackage{cases}
\usepackage{multicol}
\usepackage{multirow}
\usepackage{enumitem}
\usepackage{bm}
\usepackage{dsfont}
\usepackage{tabularx}
\usepackage{adjustbox}
\usepackage{lscape}

\setlist[enumerate]{itemindent=\dimexpr\labelwidth+\labelsep\relax,leftmargin=0pt}

\DeclareMathOperator*{\argmax}{arg\,\max}
\DeclareMathOperator*{\argmin}{arg\,\min}

\def\cdf(#1)(#2)(#3){0.5*(1+(erf((#1-#2)/(#3*sqrt(2)))))}%

\usepackage{listings}
\usepackage{lstautogobble}
\DeclareMathOperator{\aju}{end}

\DeclareMathOperator{\participants}{part} 
\DeclareMathOperator{\len}{len} 
\DeclareMathOperator{\pref}{pref} 
 
\DeclareMathOperator{\sched}{sched}

\newcommand{\N}{\mathbb{N}}

\renewcommand{\implies}{\Rightarrow}
\newcommand{\Landau}{\mathcal{O}}

\usepackage{fancyhdr}
\title{Improving Multi-agent Coordination by Learning to Estimate Contention}

\author{
Panayiotis Danassis\and
Florian Wiedemair\And
Boi Faltings
\affiliations
Artificial Intelligence Laboratory, \'Ecole Polytechnique F\'ed\'erale de Lausanne (EPFL), Switzerland
\emails
\{firstname.lastname\}@epfl.ch
}

\begin{document}

\maketitle
\thispagestyle{fancy}

\begin{abstract}
	We present a multi-agent learning algorithm, ALMA-Learning, for efficient and fair allocations in \emph{large-scale} systems. We circumvent the traditional pitfalls of multi-agent learning (e.g., the moving target problem, the curse of dimensionality, or the need for mutually consistent actions) by relying on the ALMA heuristic as a coordination mechanism for each stage game. ALMA-Learning is decentralized, observes only own action/reward pairs, requires no inter-agent communication, and achieves near-optimal ($<5\%$ loss) and fair coordination in a variety of synthetic scenarios and a real-world meeting scheduling problem. The lightweight nature and fast learning constitute ALMA-Learning ideal for \emph{on-device} deployment.
\end{abstract}

\section{Introduction} \label{Introduction}

One of the most relevant problems in multi-agent systems is finding an optimal allocation between agents, i.e., computing a maximum-weight matching, where edge weights correspond to the utility of each alternative. Many multi-agent coordination problems can be formulated as such. Example applications include role allocation (e.g., team formation \cite{GUNN201322}), task assignment (e.g., smart factories, or taxi-passenger matching \cite{danassis2019putting,varakantham2012decision}), resource allocation (e.g., parking/charging spaces for autonomous vehicles \cite{geng2013new}), etc. What follows is \emph{applicable to any such scenario}, but for concreteness we focus on the \emph{assignment problem} (bipartite matching), one of the most fundamental combinatorial optimization problems \cite{munkres1957algorithms}.

A significant challenge for any algorithm for the assignment problem emerges from the nature of real-world applications, which are often \emph{distributed} and \emph{information-restrictive}. Sharing plans, utilities, or preferences creates high overhead, and there is often a lack of responsiveness and/or communication between the participants \cite{AAAI10-adhoc}. Achieving fast convergence and high efficiency in such information-restrictive settings is extremely challenging.

A recently proposed heuristic (ALMA \cite{ijcai201931}) was specifically designed to address the aforementioned challenges. ALMA is \emph{decentralized}, completely uncoupled (agents are only aware of their \emph{own history}), and requires \emph{no communication} between the agents. Instead, agents make decisions locally, based on the contest for resources that \emph{they} are interested in, and the agents that are interested in the \emph{same} resources. As a result, in the realistic case where each agent is interested in a \emph{subset} (of fixed size) of the total resources, ALMA's convergence time is \emph{constant} in the total problem size. This condition holds by default in many real-world applications (e.g., resource allocation in urban environments), since agents only have a local (partial) knowledge of the world, and there is typically a cost associated with acquiring a resource. This \emph{lightweight} nature of ALMA coupled with the \emph{lack of inter-agent communication}, and the \emph{highly efficient allocations} \cite{danassis2019putting,ijcai201931,danassis2020differential}, make it ideal for an \emph{on-device} solution for large-scale intelligent systems (e.g., IoT devices, smart cities and intelligent infrastructure, industry 4.0, autonomous vehicles, etc.).

Despite ALMA's high performance in a variety of domains, it remains a heuristic; i.e., sub-optimal by nature. In this work, we introduce a learning element (ALMA-Learning) that allows to quickly close the gap in social welfare compared to the optimal solution, while simultaneously increasing the fairness of the allocation. Specifically, in ALMA, while contesting for a resource, each agent will back-off with probability that depends on their \emph{own utility loss} of switching to some alternative. ALMA-Learning improves upon ALMA by allowing agents to learn the chances that they will actually obtain the alternative option they consider when backing-off, which helps guide their search.

ALMA-Learning is applicable in \emph{repeated allocation} games (e.g., self organization of intelligent infrastructure, autonomous mobility systems, etc.), but can be also applied as a \emph{negotiation protocol} in one-shot interactions, where agents can simulate the learning process offline, before making their final decision. A motivating \emph{real-world} application is presented in Section \ref{Evaluation: Test Case 2}, where ALMA-Learning is applied to solve a large-scale meeting scheduling problem.

\subsection{Our Contributions}

\begin{enumerate} [label=\textbf{(\arabic*)}]
	\item We introduce \textbf{ALMA-Learning}, a distributed algorithm for large-scale multi-agent coordination, focusing on \emph{scalability} and \emph{on-device} deployment in real-world applications.
	\item We prove that ALMA-Learning \emph{converges}.
	\item We provide a thorough evaluation in a variety of synthetic benchmarks and a real-world meeting scheduling problem. In all of them ALMA-Learning is able to quickly (as little as $64$ training steps) reach allocations of high social welfare (less than $5\%$ loss) and fairness (up to almost $10\%$ lower inequality compared to the best performing baseline).
\end{enumerate}

\subsection{Discussion and Related Work} \label{Related Work}

Multi-agent coordination can usually be formulated as a matching problem. Finding a maximum weight matching is one of the best-studied combinatorial optimization problems (see \cite{su2015algorithms,lovasz2009matching}). There is a plethora of polynomial time algorithms, with the \emph{Hungarian} algorithm \cite{kuhn1955hungarian} being the most prominent centralized one for the bipartite variant (i.e., the assignment problem). In real-world problems, a centralized coordinator is not always available, and if so, it has to know the utilities of all the participants, which is often not feasible. Decentralized algorithms (e.g., \cite{giordani2010distributed}) solve this problem, yet they require polynomial computational time and polynomial number of messages -- such as cost matrices \cite{7991447}, pricing information \cite{zavlanos2008distributed}, or a basis of the LP \cite{burger2012distributed}, etc. (see also \cite{Kuhn2016LCL29061422742012,Elkin2004DAS10549161054931} for general results in distributed approximability under only local information/computation).

While the problem has been `solved' from an algorithmic perspective -- having both centralized and decentralized polynomial algorithms -- it is not so from the perspective of multi-agent systems, for two key reasons: (1) complexity, and (2) communication. The proliferation of intelligent systems will give rise to \emph{large-scale}, \emph{multi-agent} based technologies. Algorithms for maximum-weight matching, whether centralized or distributed, have runtime that increases with the total problem size, even in the realistic case where agents are interested in a small number of resources. Thus, they can only handle problems of some bounded size. Moreover, they require a significant amount of inter-agent communication. As the number and diversity of autonomous agents continue to rise, differences in origin, communication protocols, or the existence of legacy agents will bring forth the need to collaborate without any form of explicit communication \cite{AAAI10-adhoc}. Most importantly though, communication between participants (sharing utility tables, plans, and preferences) creates high overhead. On the other hand, under reasonable assumptions about the preferences of the agents, ALMA's runtime is \emph{constant} in the total problem size, while requiring no message exchange (i.e., no communication network) between the participating agents. The proposed approach, ALMA-Learning, \emph{preserves} the aforementioned two properties of ALMA.

From the perspective of Multi-Agent Learning (MAL), the problem at hand falls under the paradigm of multi-agent reinforcement learning, where for example it can be modeled as a Multi-Armed Bandit (MAB) problem \cite{auer2002nonstochastic}, or as a Markov Decision Process (MDP) and solved using a variant of Q-Learning \cite{Busoniu2008CSM22204332221106}. In MAB problems an agent is given a number of arms (resources) and at each time-step has to decide which arm to pull to get the maximum expected reward. In Q-learning agents solve Bellman's optimality equation \cite{bellman2013dynamic} using an iterative approximation procedure so as to maximize some notion of expected cumulative reward. Both approaches have arguably been designed to operate in a more challenging setting, thus making them susceptible to many pitfalls inherent in MAL. For example, there is no stationary distribution, in fact, rewards depend on the joint action of the agents and since all agents learn simultaneously, this results in a moving-target problem. Thus, there is an inherent need for coordination in MAL algorithms, stemming from the fact that the effect of an agent's action depends on the actions of the other agents, i.e. actions must be mutually consistent to achieve the desired result. Moreover, the curse of dimensionality makes it difficult to apply such algorithms to large scale problems. ALMA-Learning solves both of the above challenges \emph{by relying on ALMA as a coordination mechanism for each stage of the repeated game}. Another fundamental difference is that the aforementioned algorithms are designed to tackle the exploration/exploitation dilemma. A bandit algorithm for example will constantly explore, even if an agent has acquired his most preferred alternative. In matching problems, though, agents know (or have an estimate of) their own utilities. ALMA-Learning in particular, requires the knowledge of personal preference ordering and pairwise differences of utility (which are far easier to estimate than the exact utility table). The latter gives a great advantage to ALMA-Learning, since agents do not need to continue exploring after successfully claiming a resource, which stabilizes the learning process.

\section{Proposed Approach: ALMA-Learning} \label{Proposed Approach}


\subsection{The Assignment Problem} \label{The Assignment Problem}

The assignment problem refers to finding a maximum weight matching in a weighted bipartite graph, $\mathcal{G} = \left\{ \mathcal{N} \cup \mathcal{R}, \mathcal{V} \right\}$. In the studied scenario, $\mathcal{N} = \{1, \dots, N\}$ agents compete to acquire $\mathcal{R} = \{1, \dots, R\}$ resources. The weight of an edge $(n, r) \in \mathcal{V}$ represents the utility ($u_n(r) \in [0, 1]$) agent $n$ receives by acquiring resource $r$. Each agent can acquire at most one resource, and each resource can be assigned to at most one agent. The goal is to maximize the social welfare (sum of utilities), i.e., $\max_{\mathbf{x} \geq 0} \sum_{(n,r) \in \mathcal{V}} u_n(r) x_{n, r}$, where $\mathbf{x} = (x_{1, 1}, \dots, x_{N, R})$, subject to $\sum_{r | (n,r) \in \mathcal{V}} x_{n, r} = 1, \forall n \in \mathcal{N}$, and $\sum_{n | (n,r) \in \mathcal{V}} x_{n, r} = 1, \forall r \in \mathcal{R}$.

\subsection{Learning Rule} \label{Learning Rule}

We begin by describing (a slightly modified version of) the ALMA heuristic of \cite{ijcai201931}, which is used as a subroutine by ALMA-Learning. The pseudo-codes for ALMA and ALMA-Learning are presented in Algorithms \ref{algo: alma} and \ref{algo: alma-learning}, respectively. Both ALMA and ALMA-Learning are run \emph{independently and in parallel by all the agents} (to improve readability, we have omitted the subscript $_n$).

We make the following two assumptions: First, we assume (possibly noisy) knowledge of personal utilities by each agent. Second, we assume that agents can observe feedback from their environment to inform collisions and detect free resources. It could be achieved by the use of \emph{sensors}, or by a \emph{single bit} (0 / 1) feedback from the resource (note that these messages would be between the requesting agent and the resource, not between the participating agents themselves).

For both ALMA, and ALMA-Learning, each agent sorts his available resources (possibly $\mathcal{R}^n \subseteq \mathcal{R}$) in decreasing utility ($r_0, \dots,$ $r_i, \dots, r_{R^n - 1}$) under his preference ordering $\prec_n$.

\subsubsection{\textbf{ALMA: ALtruistic MAtching Heuristic}} 

ALMA converges to a resource through repeated trials. Let $\mathcal{A} = \{Y, A_{r_1}, \dots, A_{r_{R^n}}\}$ denote the set of actions, where $Y$ refers to yielding, and $A_r$ refers to accessing resource $r$, and let $g$ denote the agent's strategy. As long as an agent has not acquired a resource yet, at every time-step, there are two possible scenarios: If $g = A_r$ (strategy points to resource $r$), then agent $n$ attempts to acquire that resource. If there is a collision, the colliding parties back-off (set $g \leftarrow Y$) with some probability. Otherwise, if $g = Y$, the agent chooses another resource $r$ for monitoring. If the resource is free, he sets $g \leftarrow A_r$.

The back-off probability ($P(\cdot)$) is computed individually and locally based on each agent's expected loss. If more than one agent compete for resource $r_i$ (step 8 of Alg. \ref{algo: alma}), each of them will back-off with probability that depends on their expected utility loss. The expected loss array is computed by ALMA-Learning and provided as input to ALMA. The actual back-off probability can be computed with any monotonically decreasing function on $loss$ (see \cite{ijcai201931}). In this work we use $P(loss) = f(loss)^\beta$, where $\beta$ controls the aggressiveness (willingness to back-off), and 

\small
\begin{equation} \label{Eq: loss}
	f(loss) =
	\begin{cases}
		1 - \epsilon, & \text{ if } loss \leq \epsilon \\
		\epsilon, & \text{ if } 1 - loss \leq \epsilon \\
		1 - loss, & \text{ otherwise}
	\end{cases}
\end{equation}
\normalsize

Agents that do not have good alternatives will be less likely to back-off and vice versa. The ones that do back-off select an alternative resource and examine its availability. The resource selection is performed in sequential order, starting from the most preferred resource (see step 3 of Alg. \ref{algo: alma}).

\begin{table}[t!]
	\centering
	\begin{subfigure}[t]{0.48\linewidth}
		\centering
		\includegraphics[trim={0em 0em 0em 0em}, clip]{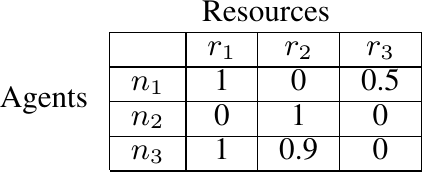}
		\caption{}
		\label{tb: adversarial example 1}
	\end{subfigure}
	~
	\begin{subfigure}[t]{0.48\linewidth}
		\centering
		\includegraphics[trim={0em 0em 0em 0em}, clip]{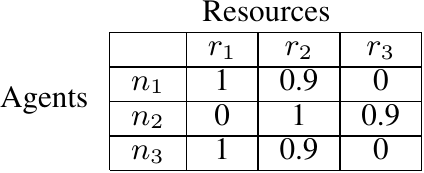}
		\caption{}
		\label{tb: adversarial example 2}
	\end{subfigure}%
	\caption{Adversarial examples: (\ref{tb: adversarial example 1}) \emph{Inaccurate loss estimate}. Agent $n_3$ backs-off with high probability when contesting for resource $r_1$ assuming a good alternative, only to find resource $r_2$ occupied.\\(\ref{tb: adversarial example 2}) \emph{Inaccurate reward expectation}. Agents $n_1$ and $n_3$ always start by attempting to acquire resource $r_1$, reasoning that it is the most preferred one, yet each of them only wins $r_1$ half of the time.}
	\label{tb: adversarial examples SW}
\end{table}

\paragraph{Sources of Inefficiency} ALMA is a heuristic, i.e., sub-optimal by nature. It is worth understanding the sources of inefficiency, which in turn motivated ALMA-Learning. To do so, we provide a couple of adversarial examples.

In the original ALMA algorithm, all agents start attempting to claim their most preferred resource, and back-off with probability that depends on their loss of switching to the immediate next best resource. Specifically, in the simplest case, the probability to back-off when contesting resource $r_i$ would be given by $P(loss(i)) = 1 - loss(i)$, where $loss(i) = u_n(r_i) - u_n(r_{i + 1})$ and $r_{i + 1}$ is the next best resource according to agent $n$'s preferences $\prec_n$.

The first example is given in \autoref{tb: adversarial example 1}. Agent $n_3$ backs-off with high probability (higher than agent $n_1$) when contesting for resource $r_1$ assuming a good alternative, only to find resource $r_2$ occupied. Thus, $n_3$ ends up matched with resource $r_3$. The social welfare of the final allocation is 2, which is $20\%$ worse than the optimal (where agents $n_1, n_2, n_3$ are matched with resources $r_3, r_2, r_1$, respectively, achieving a social welfare of 2.5). ALMA-Learning solves this problem by learning an empirical estimate of the loss an agent will incur if he backs-off from a resource. In this case, agent $n_3$ will learn that his loss is not $1 - 0.9 = 0.1$, but actually $1 - 0 = 1$, and thus will not back-off in subsequent stage games, resulting in an optimal allocation.

In another example (\autoref{tb: adversarial example 2}, agents $n_1$ and $n_3$ always start by attempting to acquire resource $r_1$, reasoning that it is the most preferred one. Yet, in a repeated game, each of them only wins $r_1$ half of the time (for a social welfare 2, which is $28.5\%$ worse than the optimal 2.8), thus, in expectation, resource $r_1$ has utility $0.5$. ALMA-Learning solves this by learning an empirical estimate of the reward of each resource. In this case, after learning, either agent $n_1$ or $n_3$ (or both), will start from resource $r_2$. Agent $n_2$ will back-off since he has a good alternative, and the result will be the optimal allocation where agents $n_1, n_2, n_3$ are matched with resources $r_2, r_3, r_1$ (or $r_1, r_3, r_2$), respectively.

\subsubsection{\textbf{ALMA-Learning}}

ALMA-Learning uses ALMA as a sub-routine, specifically as a coordination mechanism for each stage of the repeated game. Over time, ALMA-Learning learns which resource to select first ($r_{start}$) when running ALMA, and an accurate empirical estimate on the loss the agent will incur by backing-off ($loss[]$). By learning these two values agents take more informed decisions, specifically: (1) If an agent often loses the contest of his starting resource, the expected reward of that resource will decrease, thus in the future the agent will switch to an alternative starting resource, and (2) if an agent backs-off from contesting resource $r$ expecting low loss, only to find that all his high utility alternatives are already occupied, then his expected loss of resource $r$ ($loss[r]$) will increase, making him more reluctant to back-off in some future stage game. In more detail, ALMA-Learning learns and maintains the following information\footnote{We have omitted the subscript $_n$ from all the variables and arrays, but every agent maintains their own estimates.}:

\textbf{(i)} $rewardHistory[R][L]$: A 2D array. For each $r \in \mathcal{R}$ it maintains the $L$ most recent reward values received by agent $n$, i.e., the $L$ most recent $u_n(r_{won})$, where $r_{won} \leftarrow ALMA(r, loss[])$. See line 11 of Alg. \ref{algo: alma-learning}. The array is initialized to the utility of each resource (line 3 of Alg. \ref{algo: alma-learning}).

\textbf{(ii)} $reward[R]$: A 1D array. For each $r \in \mathcal{R}$ it maintains an empirical estimate on the expected reward received by starting at resource $r$ and continue playing according to Alg. \ref{algo: alma}. It is computed by averaging the reward history of the resource, i.e., $\forall r \in \mathcal{R}: reward[r] \leftarrow rewardHistory[r].getMean()$. See line 12 of Alg. \ref{algo: alma-learning}.

\textbf{(iii)} $loss[R]$: A 1D array. For each $r \in \mathcal{R}$ it maintains an empirical estimate on the loss in utility agent $n$ incurs if he backs-off from the contest of resource $r$. The loss of each resource $r$ is initialized to $loss[r] \leftarrow u_n(r) - u_n(r_{next})$, where $r_{next}$ is the next most preferred resource to $r$, according to agent $n$'s preferences $\prec_n$ (see line 5 of Alg. \ref{algo: alma-learning}). Subsequently, for every stage game, agent $n$ starts by selecting resource $r_{start}$, and ends up winning resource $r_{won}$. The loss of $r_{start}$ is then updated according to the following averaging process, where $\alpha$ is the learning rate:
\small
\begin{equation*}
	loss[r_{start}] \leftarrow (1 - \alpha) loss[r_{start}] + \alpha \left( u(r_{start}) - u(r_{won}) \right)
\end{equation*}
\normalsize

Finally, the last condition (lines 17-18 of Alg. \ref{algo: alma-learning}) ensures that agents who have acquired resources of high preference stop exploring, thus stabilizing the learning process.


\begin{algorithm}[!t]
	\small
	\caption{ALMA: Altruistic Matching Heuristic.} \label{algo: alma}
	\begin{algorithmic}[1]
		\Require Sort resources ($\mathcal{R}^n \subseteq \mathcal{R}$) in decreasing order of utility $r_0, \dots, r_{R^n - 1}$ under $\prec_n$

		\Procedure{ALMA}{$r_{start}$, $loss[R]$}
			\State Initialize $g \leftarrow A_{r_{start}}$
			\State Initialize $current \leftarrow -1$
			\State Initialize $converged \leftarrow False$
			\While{!$converged$}
				\If{ $g = A_r$}
					\State Agent $n$ attempts to acquire $r$
					\If{Collision($r$)}
					\State Back-off (set $g \leftarrow Y$) with prob. $P(loss[r])$
					\Else
					\State $converged \leftarrow True$
					\EndIf
				\Else { ($g = Y$)}
					\State $current \leftarrow (current + 1)$ mod $R$
					\State Agent $n$ monitors $r \leftarrow r_{current}$.
					\If{Free($r$)} set $g \leftarrow A_r$
					\EndIf
				\EndIf
			\EndWhile
			\State
			\Return $r$, such that $g = A_r$
		\EndProcedure
	\end{algorithmic}
	\normalsize
\end{algorithm}

\subsection{Convergence} \label{Convergecne}

Convergence of ALMA-Learning does not translate to a fixed allocation at each stage game. The system has converged when agents no longer switch their starting resource, $r_{start}$. The final allocation of each stage game is controlled by ALMA, which means that even after convergence there can be contest for a resource, i.e., having more than one agent selecting the same starting resource. As we will demonstrate later, this translates to fairer allocations, since agents with similar preferences can alternate between acquiring their most preferred resource.

\begin{theorem} \label{theorem}
	There exists time-step $t_{conv}$ such that $\forall t > t_{conv}: r_{start}^n(t) = r_{start}^n(t_{conv})$, where $r_{start}^n(t)$ denotes the starting resource $r_{start}$ of agent $n$ at the stage game of time-step $t$.
\end{theorem}

\begin{proof} (Sketch; see Appendix \ref{Appendix: Convergecne})
	Theorem 2.1 of \cite{ijcai201931} proves that ALMA (called at line 9 of Alg. \ref{algo: alma-learning}) converges in polynomial time (in fact, under some assumptions, it converges in constant time, i.e., \emph{each stage game converges in constant time}). In ALMA-Learning agents switch their starting resource only when the expected reward for the current starting resource drops below the best alternative one, i.e., for an agent to switch from $r_{start}$ to $r_{start}'$, it has to be that $reward[r_{start}] < reward[r_{start}']$. Given that utilities are bounded in $[0, 1]$, there is a maximum, finite number of switches until $reward_n[r] = 0, \forall r \in \mathcal{R}, \forall n \in \mathcal{N}$. In that case, the problem is equivalent to having $N$ balls thrown randomly and independently into $N$ bins (since $R = N$). Since both $R, N$ are finite, the process will result in a distinct allocation in finite steps with probability 1.
\end{proof}

\begin{algorithm}[!t]
	\small
	\caption{ALMA-Learning} \label{algo: alma-learning}
	\begin{algorithmic}[1]
		\Require Sort resources ($\mathcal{R}^n \subseteq \mathcal{R}$) in decreasing order of utility $r_0, \dots, r_{R^n - 1}$ under $\prec_n$
		\Require $rewardHistory[R][L]$, $reward[R]$, $loss[R]$

		\Procedure{ALMA-Learning}{}
			\ForAll{$r \in \mathcal{R}$} \Comment{Initialization}
				\State $rewardHistory[r]$.add($u(r)$)
				\State $reward[r] \leftarrow rewardHistory[r]$.getMean()
				\State $loss[r] \leftarrow u(r) - u(r_{next})$
			\EndFor
			\State $r_{start} \leftarrow \argmax_r reward[r]$
			\State
			\For{$t \in [1, \dots, T]$} \Comment $T$: Time horizon
				\State $r_{won} \leftarrow ALMA(r_{start}, loss[])$ \Comment{Run ALMA}
				\State
				\State $rewardHistory[r_{start}]$.add($u(r_{won})$) 
				\State $reward[r_{start}] \leftarrow rewardHistory[r_{start}]$.getMean()
				\If{$u(r_{start}) - u(r_{won}) > 0$}
				\State $loss[r_{start}] \leftarrow$
				\State $\quad (1 - \alpha) loss[r_{start}] + \alpha \left( u(r_{start}) - u(r_{won}) \right)$
				\EndIf
				\State
				\If{$r_{start} \neq r_{won}$}
					\State $r_{start} \leftarrow \argmax_r reward[r]$
				\EndIf
			\EndFor
		\EndProcedure
	\end{algorithmic}
	\normalsize
\end{algorithm}

\section{Evaluation} \label{Evaluation}

We evaluate ALMA-Learning in a variety of synthetic benchmarks and a meeting scheduling problem based on \emph{real} data from \cite{romano01}. Error bars representing one standard deviation (SD) of uncertainty.

For brevity and to improve readability, we only present the most relevant results in the main text. We refer the interested reader to the appendix for additional results for both Sections \ref{Evaluation: Test Case 1}, \ref{Evaluation: Test Case 2}, implementation details and hyper-parameters, and a detailed model of the meeting scheduling problem.

\paragraph{Fairness}

The usual predicament of efficient allocations is that they assign the resources only to a fixed subset of agents, which leads to an unfair result. Consider the simple example of \autoref{tb: fairness example}. Both ALMA (with higher probability) and any optimal allocation algorithm will assign the coveted resource $r_1$ to agent $n_1$, while $n_3$ will receive utility 0. But, using ALMA-Learning, agents $n_1$ and $n_3$ will update their expected loss for resource $r_1$ to 1, and randomly acquire it between stage games, increasing fairness. Recall that convergence for ALMA-Learning does not translate to a fixed allocation at each stage game. To capture the fairness of this `mixed' allocation, we report the average fairness on 32 evaluation time-steps that follow the training period.

To measure fairness, we used the \emph{Gini coefficient} \cite{gini1912variabilita}. An allocation $\mathbf{x} = (x_1, \dots, x_N) ^\top$ is fair iff $\mathds{G}(\mathbf{x}) = 0$, where:
$\mathds{G}(\mathbf{x}) = (\sum^{N}_{n = 1} \sum^{N}_{n' = 1} \left| x_n - x_{n'} \right|)     /      (2 N \sum^{N}_{n = 1}  x_n)$




\begin{table}[t!]
	\small
	\centering
	\begin{tabular}{ccccc}
	                                             & \multicolumn{4}{c}{Resources}                                                                                                       \\ \cline{2-5} 
	\multicolumn{1}{c|}{\multirow{4}{*}{Agents}} & \multicolumn{1}{c|}{}      & \multicolumn{1}{c|}{$r_1$} & \multicolumn{1}{c|}{$r_2$} & \multicolumn{1}{c|}{$r_3$}                   \\ \cline{2-5} 
	\multicolumn{1}{c|}{}                        & \multicolumn{1}{c|}{$n_1$} & \multicolumn{1}{c|}{1}     & \multicolumn{1}{c|}{0.5}   & \multicolumn{1}{c|}{0}                      \\ \cline{2-5} 
	\multicolumn{1}{c|}{}                        & \multicolumn{1}{c|}{$n_2$} & \multicolumn{1}{c|}{0}     & \multicolumn{1}{c|}{1}     & \multicolumn{1}{c|}{0}                       \\ \cline{2-5} 
	\multicolumn{1}{c|}{}                        & \multicolumn{1}{c|}{$n_3$} & \multicolumn{1}{c|}{1}     & \multicolumn{1}{c|}{0.75}   & \multicolumn{1}{c|}{$\epsilon \rightarrow 0$}\\ \cline{2-5} 
	\end{tabular}
	\normalsize
	\caption{Adversarial example: Unfair allocation. Both ALMA (with higher probability) and any optimal allocation algorithm will assign the coveted resource $r_1$ to agent $n_1$, while $n_3$ will receive utility 0.}
	\label{tb: fairness example}
\end{table}

\subsection{Test Case \#1: Synthetic Benchmarks} \label{Evaluation: Test Case 1}

\paragraph{Setting}

We present results on three benchmarks:
\begin{enumerate} [label=\textbf{(\alph*)}]
	\item \emph{Map}: Consider a Cartesian map on which the agents and resources are randomly distributed. The utility of agent $n$ for acquiring resource $r$ is proportional to the inverse of their distance, i.e., $u_n(r) = 1 / d_{n,r}$. Let $d_{n,r}$ denote the Manhattan distance. We assume a grid length of size $\sqrt{4 \times N}$.
	\item \emph{Noisy Common Utilities}: This pertains to an anti-coordination scenario, i.e., competition between agents with similar preferences. We model the utilities as: $\forall n, n' \in \mathcal{N}, |u_n(r) - u_{n'}(r)| \leq \text{noise}$, where the noise is sampled from a zero-mean Gaussian distribution, i.e., $\text{noise} \sim \mathcal{N}(0, \sigma^2)$.
	\item \emph{Binary Utilities}: This corresponds to each agent being indifferent to acquiring any resource amongst his set of desired resources, i.e., $u_n(r)$ is randomly assigned to 0 or 1.
\end{enumerate}

\paragraph{Baselines}

We compare against: (i) the \emph{Hungarian algorithm} \cite{kuhn1955hungarian}, which computes a maximum-weight matching in a bipartite graphs, (ii) \emph{ALMA} \cite{ijcai201931}, and (iii) the \emph{Greedy} algorithm, which goes through the agents randomly, and assigns them their most preferred, unassigned resource.



\paragraph{Results} \label{Results: Test-Case 1}

We begin with the loss in social welfare. Figure \ref{fig: map_tt_512_RDSW} depicts the results for the Map test-case, while Table \ref{tb: social welfare} aggregates all three test-cases\footnote{For the Noisy Common Utilities test-case, we report results for $\sigma = 0.1$; which is the worst performing scenario for ALMA-Learning. Similar results were obtained for $\sigma = 0.2$ and $\sigma = 0.4$.}. ALMA-Leaning reaches \emph{near-optimal} allocations (\emph{less than $2.5\%$ loss}), in most cases in just $32-512$ training time-steps. The exception is the Noisy Common Utilities test-case, where the training time was slightly higher. Intuitively we believe that this is because ALMA already starts with a near optimal allocation (especially for $R > 256$), and given the high similarity on the agent's utility tables (especially for $\sigma = 0.1$), it requires a lot of fine-tuning to improve the result.


Moving on to fairness, ALMA-Leaning achieves the \emph{most fair allocations in all of the test-cases}. As an example, Figure \ref{fig: map_tt_512_GI} depicts the Gini coefficient for the Map test-case. ALMA-Learning's Gini coefficient is $-18\%$ to $-90\%$ lower on average (across problem sizes) than ALMA's, $-24\%$ to $-93\%$ lower than Greedy's, and $-0.2\%$ to $-7\%$ lower than Hungarian's.

\begin{figure}[t!]
	\centering
	\begin{subfigure}[t]{0.48\linewidth}
		\centering
		\includegraphics[width = 1 \linewidth, trim={0em 0em 0em 0em}, clip]{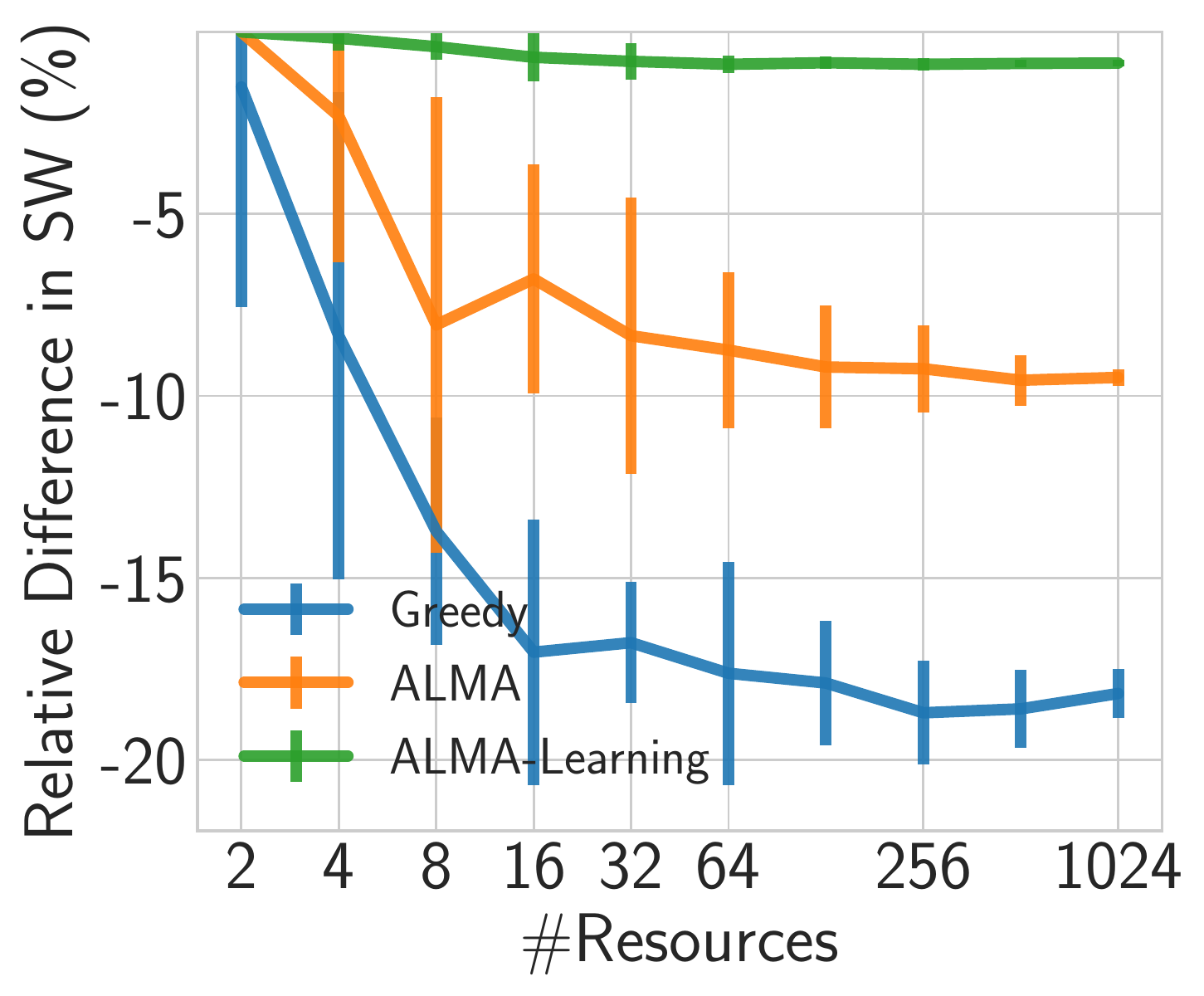}
		\caption{Relative Difference in SW}
		\label{fig: map_tt_512_RDSW}
	\end{subfigure}
	~
	\begin{subfigure}[t]{0.48\linewidth}
		\centering
		\includegraphics[width = 1 \linewidth, trim={0em 0em 0em 0em}, clip]{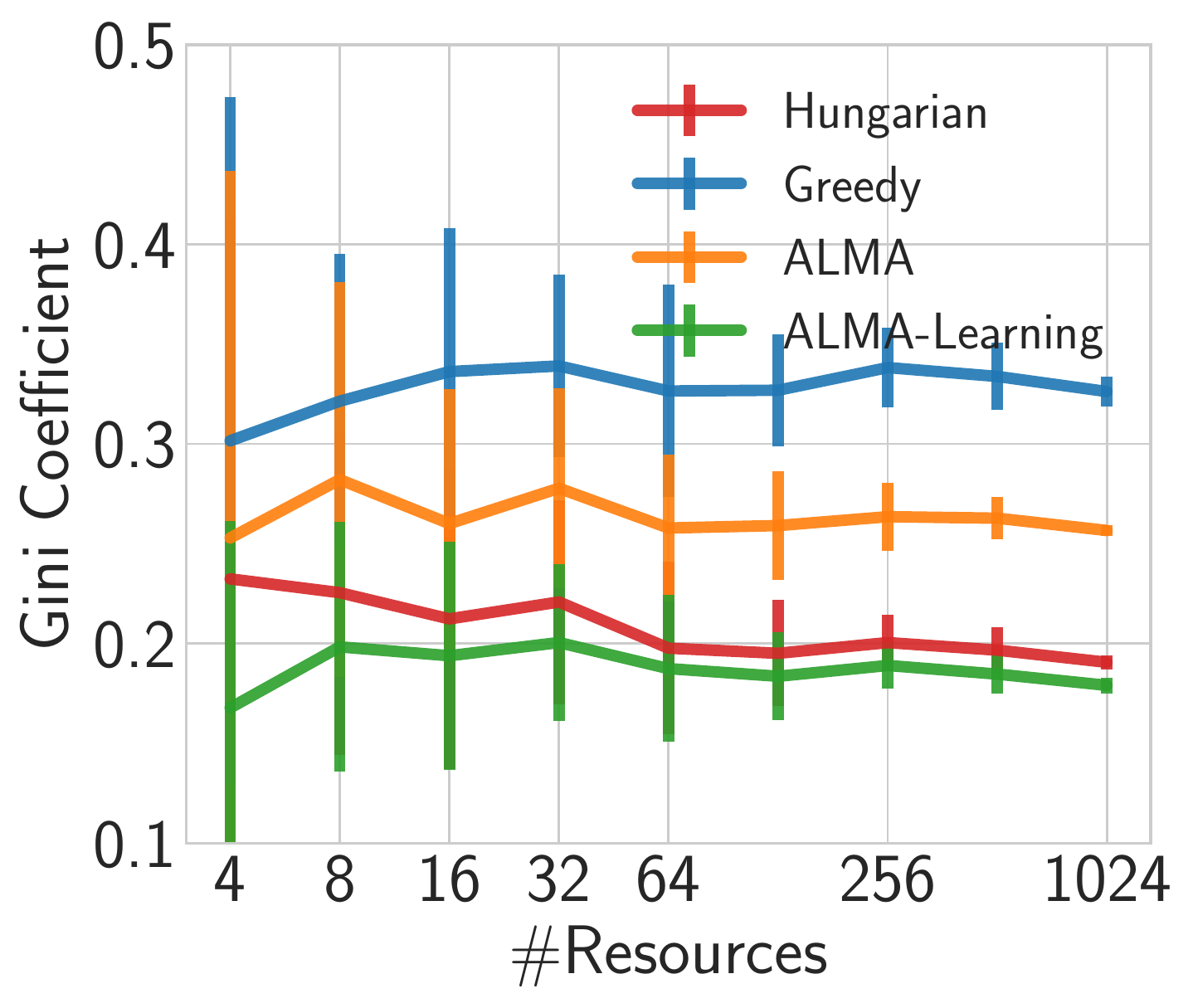}
		\caption{Gini Index (lower is better)}
		\label{fig: map_tt_512_GI}
	\end{subfigure}
	\caption{Map test-case. Results for increasing number of resources ($[2, 1024]$, $x$-axis in log scale), and $N = R$. ALMA-Learning was trained for 512 time-steps.}
	\label{fig: results test-case 1}
\end{figure}

\begin{table}[t!]
	\centering
	\caption{Range of the average loss ($\%$) in social welfare compared the (centralized) optimal for the three different benchmarks.}
	\label{tb: social welfare}
	\resizebox{\linewidth}{!}{%
	\begin{tabular}{@{}rccc@{}}
	\toprule
	\multicolumn{1}{c}{} & Greedy             & ALMA               & \textbf{ALMA-Learning} \\ \midrule
	(a) Map              & $1.51\% - 18.71\%$ & $0.00\% - 9.57\%$  & $0.00\% - 0.89\%$ \\
	(b) Noisy            & $8.13\% - 12.86\%$ & $2.96\% - 10.58\%$ & $1.34\% - 2.26\%$ \\
	(c) Binary           & $0.10\% - 14.70\%$ & $0.00\% - 16.88\%$ & $0.00\% - 0.39\%$ \\ \bottomrule
	\end{tabular}%
	}
\end{table}

\subsection{Test Case \#2: Meeting Scheduling} \label{Evaluation: Test Case 2}

\begin{figure*}[t!]
	\centering
	\begin{subfigure}[t]{0.26\textwidth}
		\centering
		\includegraphics[width = 1 \linewidth, trim={0em 0em 18em 0em}, clip]{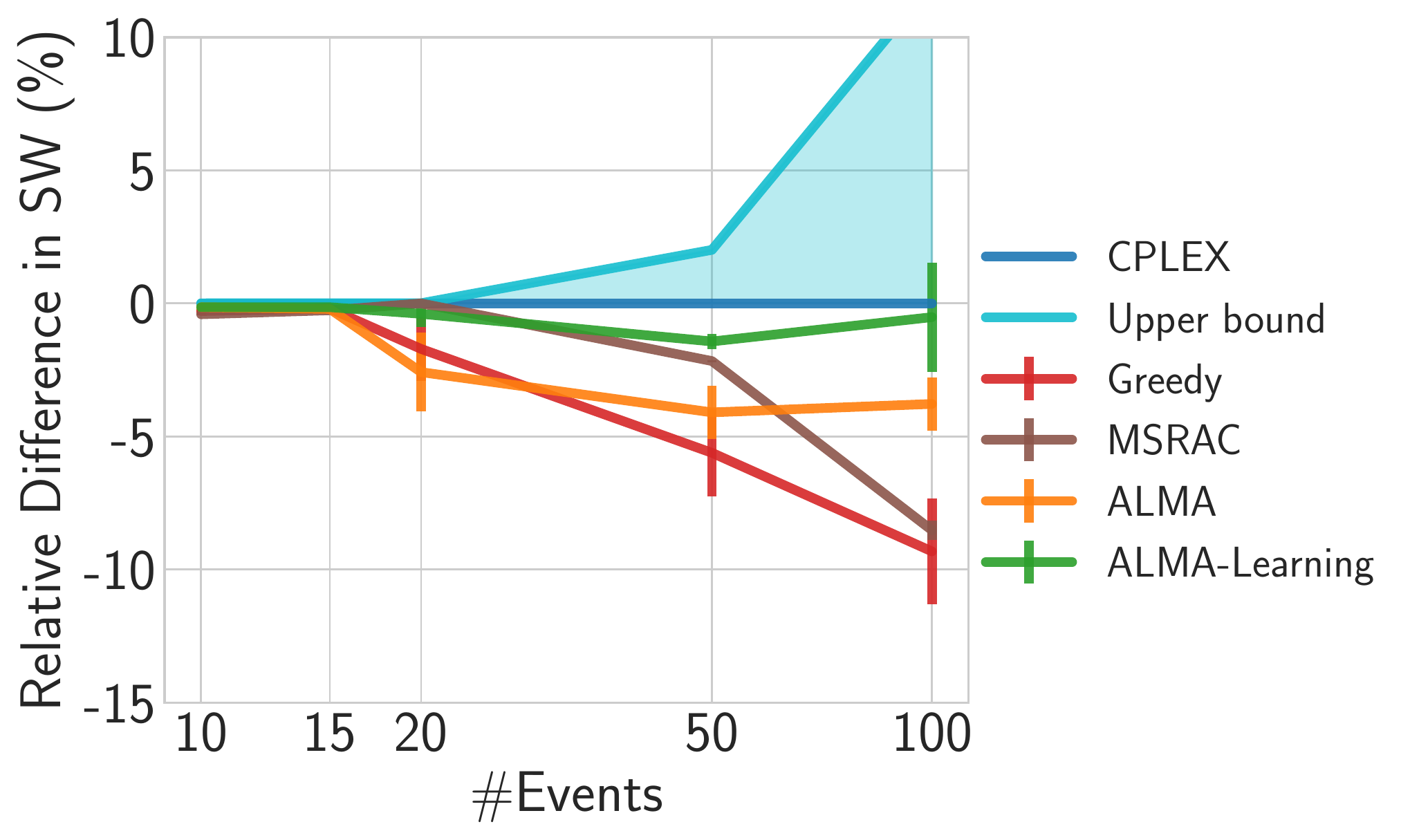}
		\caption{Relative Difference in SW}
		\label{fig: meeting_scheduling_sw}
	\end{subfigure}
	~ 
	\begin{subfigure}[t]{0.26\textwidth}
		\centering
		\includegraphics[width = 1 \linewidth, trim={0em 0em 18.1em 0em}, clip]{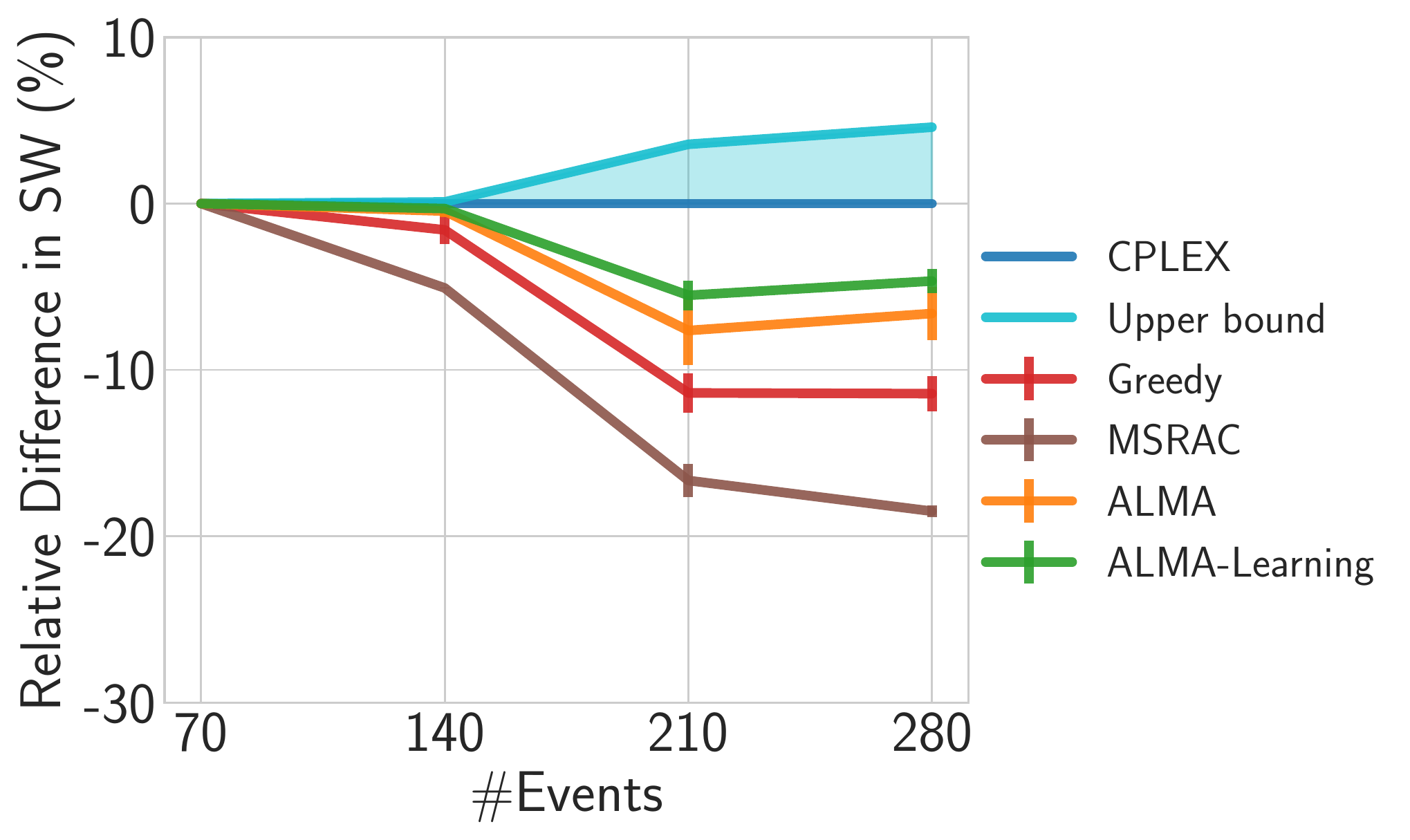}
		\caption{Relative Difference in SW}
		\label{fig: meeting_scheduling_sw_large}
	\end{subfigure}
	~ 
	\begin{subfigure}[t]{0.26\textwidth}
		\centering
		\includegraphics[width = 1 \linewidth, trim={0em 0em 18em 0em}, clip]{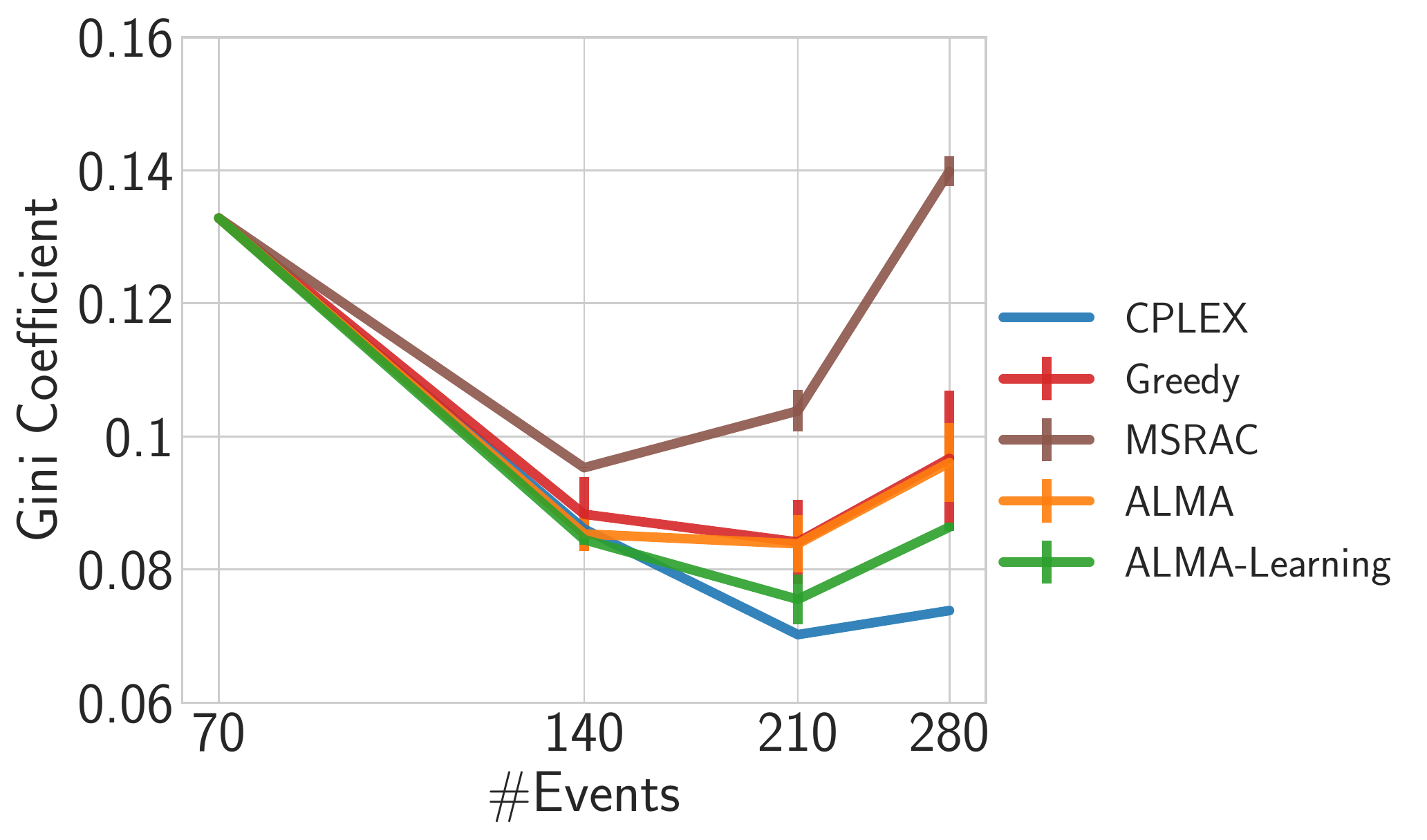}
		\caption{Gini Coefficient (lower is better)}
		\label{fig: meeting_scheduling_gi}
	\end{subfigure}
	~ 
	\begin{subfigure}[t]{0.15\textwidth}
		\centering
		\includegraphics[width = 1 \linewidth, scale=0.8, trim={41em 10.5em 0em 10em}, clip]{./general_100_GDSW.pdf}
		\label{fig: legend}
	\end{subfigure}%
	\caption{Meeting Scheduling. Results for 100 participants ($\mathcal{P}$) and increasing number of events ($x$-axis in log scale). ALMA-Learning was trained for 512 time-steps.}
	\label{fig: meeting_scheduling_results}
\end{figure*}

\paragraph{Motivation}

The problem of scheduling a large number of meetings between multiple participants is ubiquitous in everyday life \cite{nigam20,ottens12,zunino09,maheswaran04,benhassine07,10.5555/1018411.1018882,crawford2005learning,franzin02}. The advent of social media brought forth the need to schedule large-scale events, while the era of globalization and the shift to working-from-home require business meetings to account for participants with diverse preferences (e.g., different timezones).

Meeting scheduling is an inherently decentralized problem. Traditional approaches (e.g., distributed constraint optimization \cite{ottens12,maheswaran04}) can only handle a bounded, small number of meetings. Interdependences between meetings' participants can drastically increase the complexity. While there are many commercially available electronic calendars (e.g., Doodle, Google calendar, Microsoft Outlook, Apple's Calendar, etc.), none of these products is capable of autonomously scheduling meetings, taking into consideration user preferences and availability.

While the problem is inherently online, meetings can be aggregated and scheduled in batches, similarly to the approach for tackling matchings in ridesharing platforms \cite{danassis2019putting}. In this test-case, we map meeting scheduling to an allocation problem and solve it using ALMA and ALMA-Learning. This showcases an application were ALMA-Learning can be used as a negotiation protocol.

\paragraph{Modeling}

Let $\mathcal{E} = \{E_1,\dots,E_n\}$ denote the set of events and $\mathcal{P} = \{P_1,\dots,P_m\}$ the set of participants. To formulate the participation, let $\participants: \mathcal{E} \to 2^\mathcal{P}$, where $2^\mathcal{P}$ denotes the power set of $\mathcal{P}$. We further define the variables $days$ and $slots$ to denote the number of days and time slots per day of our calendar (e.g., $days = 7, slots = 24$). In order to add length to each event, we define an additional function $\len : \mathcal{E} \to \N$. Participants' utilities are given by:

\noindent
$\pref : \mathcal{E} \times \participants(\mathcal{E}) \times \{1,\dots, days\} \times \{1,\dots, slots\} \to [0,1]$.


Mapping the above to the assignment problem of Section \ref{The Assignment Problem}, we would have the set of ($day$, $slot$) tuples to correspond to $\mathcal{R}$, while each event is represented by one event agent (that aggregates the participant preferences), the set of which would correspond to $\mathcal{N}$.

\paragraph{Baselines}

We compare against four baselines: (a) We used the IBM ILOG CP optimizer \cite{laborie2018ibm} to formulate and solve the problem as a CSP\footnote{Computation time limit $20$ minutes.}. An additional benefit of this solver is that it provides an upper bound for the optimal solution (which is infeasible to compute). (b) A modified version of the MSRAC algorithm \cite{benhassine07}, and finally, (c) the greedy and (d) ALMA, as before.

\paragraph{Designing Large Test-Cases}

As the problem size grows, CPLEX's estimate on the upper bound of the optimal solution becomes too loose (see Figure \ref{fig: meeting_scheduling_sw}). To get a more accurate estimate on the loss in social welfare for larger test-cases, we designed a large-instance by combining smaller problem instances, making it easier for CPLEX to solve which in turn allowed for tighter upper bounds as well (see Figure \ref{fig: meeting_scheduling_sw_large}).

We begin by solving two smaller problem instances with a low number of events. We then combine the two in a calendar of twice the length by duplicating the preferences, resulting in an instance of twice the number of events (agents) and calendar slots (resources). Specifically, in this case we generated seven one-day long sub-instances (with 10, 20, 30 and 40 events each), and combined then into a one-week long instance with 70, 140, 210 and 280 events, respectively. The fact that preferences repeat periodically, corresponds to participants being indifferent on the day (yet still have a preference on time).

These instances are depicted in Figure \ref{fig: meeting_scheduling_sw_large} and in the last line of Table \ref{tb: meeting scheduling social welfare}.

\paragraph{Results}

Figures \ref{fig: meeting_scheduling_sw} and \ref{fig: meeting_scheduling_sw_large} depict the relative difference in social welfare compared to CPLEX for 100 participants ($|\mathcal{P}| = 100$) and increasing number of events for the regular ($|\mathcal{E}| \in [10, 100]$) and larger test-cases ($|\mathcal{E}|$ up to 280), respectively. Table \ref{tb: meeting scheduling social welfare} aggregates the results for various values of $\mathcal{P}$. ALMA-Learning is able to achieve less than $5\%$ loss compared to CPLEX, and this difference diminishes as the problem instance increases (less than $1.5\%$ loss for $|\mathcal{P}| = 100$). Finally, for the largest hand-crafted instance ($|\mathcal{P}| = 100, |\mathcal{E}| = 280$, last line of Table \ref{tb: meeting scheduling social welfare} and Figure \ref{fig: meeting_scheduling_sw_large}), ALMA-Learning loses less than $9\%$ compared to the \emph{possible upper bound} of the optimal solution.

Moving on to fairness, Figure \ref{fig: meeting_scheduling_gi} depicts the Gini coefficient for the large, hand-crafted instances ($|\mathcal{P}| = 100, |\mathcal{E}|$ up to 280). ALMA-Learning exhibits low inequality, up to $-9.5\%$ lower than ALMA in certain cases. It is worth noting, though, that the fairness improvement is not as pronounced as in Section \ref{Evaluation: Test Case 1}. In the meeting scheduling problem, all of the employed algorithms exhibit high fairness, due to the nature of the problem. Every participant has multiple meetings to schedule (contrary to only being matched to a single resource), all of which are drawn from the same distribution. Thus, as you increase the number of meetings to be scheduled, the fairness naturally improves.

\begin{table}[t!]
\centering
\caption{Range of the average loss ($\%$) in social welfare compared to the IBM ILOG CP optimizer for increasing number of participants, $\mathcal{P}$ ($|\mathcal{E}| \in [10, 100]$). The final line corresponds to the loss compared to the upper bound for the optimal solution for the large test-case with $|\mathcal{P}| = 100, |\mathcal{E}| = 280$ (Figure \ref{fig: meeting_scheduling_sw_large}).}
\label{tb: meeting scheduling social welfare}
\resizebox{\linewidth}{!}{%
\begin{tabular}{@{}rcccc@{}}
\toprule
\multicolumn{1}{c}{}  & Greedy               & MSRAC                & ALMA                & \textbf{ALMA-Learning} \\ \midrule
$|\mathcal{P}| = 20$  & $6.16 \% - 18.35 \%$ & $0.00 \% - 8.12 \%$  & $0.59 \% - 8.69 \%$ & $0.16 \% - 4.84 \%$ \\
$|\mathcal{P}| = 30$  & $1.72 \% - 14.92 \%$ & $1.47 \% - 10.81 \%$ & $0.50 \% - 8.40 \%$ & $0.47 \% - 1.94 \%$ \\
$|\mathcal{P}| = 50$  & $3.29 \% - 12.52 \%$ & $0.00 \% - 15.74 \%$ & $0.07 \% - 7.34 \%$ & $0.05 \% - 1.68 \%$ \\ 
$|\mathcal{P}| = 100$ & $0.19 \% - 9.32 \%$  & $0.00 \% - 8.52 \%$  & $0.15 \% - 4.10 \%$ & $0.14 \% - 1.43 \%$ \\ \midrule
$|\mathcal{E}| = 280$ & $0.00 \% - 15.31 \%$ & $0.00 \% - 22.07 \%$ & $0.00 \% - 10.81 \%$ & $0.00 \% - 8.84 \%$\\ \bottomrule
\end{tabular}%
}
\end{table}

\section{Conclusion} \label{Conclusion}

The next technological revolution will be interwoven to the proliferation of intelligent systems. To truly allow for scalable solutions, we need to shift from traditional approaches to multi-agent solutions, ideally run \emph{on-device}. In this paper, we present a novel learning algorithm (ALMA-Learning), which exhibits such properties, to tackle a central challenge in multi-agent systems: finding an optimal allocation between agents, i.e., computing a maximum-weight matching. We prove that ALMA-Learning converges, and provide a thorough empirical evaluation in a variety of synthetic scenarios and a real-world meeting scheduling problem. ALMA-Learning is able to quickly (in as little as 64 training steps) reach allocations of high social welfare (less than $5\%$ loss) and fairness.


\clearpage
\pagestyle{empty}
\appendix

\section*{Appendix: Contents}

In this appendix we include several details that have been omitted from the main text for the shake of brevity and to improve readability. In particular:

\begin{itemize}
	\item[-] In Section \ref{Appendix: Convergecne}, we prove Theorem \ref{theorem}.
	\item[-] In Section \ref{Appendix: Modeling of the Meeting Scheduling Problem}, we describe in detail the modeling of the meeting scheduling problem, including the problem formulation, the data generation, the modeling of the events, the participants, and the utility functions, and finally several implementation related details.
	\item[-] In Section \ref{Appendix: Evaluation}, we provide a thorough account of the simulation results -- including but not limited to omitted graphs and tables -- both for the synthetic benchmarks and the meeting scheduling problem.
\end{itemize}

\section{Proof of Theorem \ref{theorem}} \label{Appendix: Convergecne}

\begin{proof}
	Theorem 2.1 of \cite{ijcai201931} proves that ALMA (called at line 9 of Algorithm \ref{algo: alma-learning}) converges in polynomial time.

	In fact, under the assumption that each agent is interested in a subset of the total resources (i.e., $\mathcal{R}^n \subset \mathcal{R}$) and thus at each resource there is a bounded number of competing agents ($\mathcal{N}^r \subset \mathcal{N}$) Corollary 2.1.1 of \cite{ijcai201931} proves that the expected number of steps any individual agent requires to converge is independent of the total problem size (i.e., $N$ and $R$). In other words, by bounding these two quantities (i.e., we consider $R^n$ and $N^r$ to be constant functions of $N$, $R$), the convergence time of ALMA is \emph{constant} in the total problem size $N$, $R$. Thus, under the aforementioned assumptions:

	\begin{center}
		\emph{Each stage game converges in constant time.}
	\end{center}

	\noindent
	Now that we have established that the call to the ALMA procedure will return, the key observation to prove convergence for ALMA-Learning is that agents switch their starting resource only when the expected reward for the current starting resource drops below the best alternative one, i.e., for an agent to switch from $r_{start}$ to $r_{start}'$, it has to be that $reward[r_{start}] < reward[r_{start}']$. Given that utilities are bounded in $[0, 1]$, there is a maximum, finite number of switches until $reward_n[r] = 0, \forall r \in \mathcal{R}, \forall n \in \mathcal{N}$. In that case, the problem is equivalent to having $N$ balls thrown randomly and independently into $N$ bins (since $R = N$). Since both $R, N$ are finite, the process will result in a distinct allocation in finite steps with probability 1. In more detail, we can make the following arguments:

	(i) Let $r_{start}$ be the starting resource for agent $n$, and $r_{start}' \leftarrow \argmax_{r \in \mathcal{R} / \{r_{start}\}} reward[r]$. There are two possibilities. Either $reward[r_{start}] > reward[r_{start}']$ for all time-steps $t > t_{converged}$ -- i.e., $reward[r_{start}]$ can oscillate but always stays larger than $reward[r_{start}']$ -- or there exists time-step $t$ when $reward[r_{start}] < reward[r_{start}']$, and then agent $n$ switches to the starting resource $r_{start}'$.

	(ii) Only the reward of the starting resource $r_{start}$ changes at each stage game. Thus, for the reward of a resource to increase, it has to be the $r_{start}$. In other words, at each stage game that we select $r_{start}$ as the starting resource, the reward of every other resource remains (1) unchanged and (2) $reward[r] < reward[r_{start}], \forall r \in \mathcal{R} \ \{ r_{start} \}$ (except when an agent switches starting resources).

	(iii) There is a finite number of times each agent can switch his starting resource $r_{start}$. This is because $u_n(r) \in [0, 1]$ and $|u_n(r) - u_n(r')| > \delta, \forall n \in \mathcal{N}, r \in \mathcal{R}$, where $\delta$ is a small, strictly positive minimum increment value. This means that either the agents will perform the maximum number of switches until $reward_n[r] = 0, \forall r \in \mathcal{R} \forall n \in \mathcal{N}$ (which will happen in finite number of steps), or the process will have converged before that.

	(iv) If $reward_n[r] = 0, \forall r \in \mathcal{R}, \forall n \in \mathcal{N}$, the question of convergence is equivalent to having $N$ balls thrown randomly and independently into $R$ bins and asking whether you can have exactly one ball in each bin -- or in our case, where $N = R$, have no empty bins. The probability of bin $r$ being empty is $\left( \frac{R - 1}{R} \right) ^ N$, i.e., being occupied is $1 - \left( \frac{R - 1}{R} \right) ^ N$. The probability of all the bins to be occupied is $\left( 1 - \left( \frac{R - 1}{R} \right) ^ N \right)^R$. The expected number of trials until this event occurs is $1 / \left( 1 - \left( \frac{R - 1}{R} \right) ^ N \right)^R$, which is finite, for finite $N, R$.
\end{proof}

\subsection{Complexity}

ALMA-Learning is an anytime algorithm. At each training time-step, we run ALMA once. Thus, the computational complexity is bounded by $T$ times the bound for ALMA, where $T$ denotes the number of training time-steps (see Equation \ref{Eq: convergence bound system}, where $N$ and $R$ denote the number of agents and resources, respectively, $p^* = f(loss^*)$, and $loss^*$ is given by the Equation \ref{Eq: loss* system}).

\small
\begin{equation} \label{Eq: convergence bound system}
	\mathcal{O}\left( T R \frac{2 - p^*}{2 (1 - p^*)} \left(\frac{1}{p^*} \log N + R \right) \right)
\end{equation}
\normalsize

\small
\begin{equation} \label{Eq: loss* system}
	loss^* = \underset{loss_n^r}{\argmin} \left( \underset{r \in \mathcal{R}, n \in \mathcal{N}}{\min}(loss_n^r), 1 - \underset{r \in \mathcal{R}, n \in \mathcal{N}}{\max}(loss_n^r) \right)
\end{equation}
\normalsize

\section{Modeling of the Meeting Scheduling Problem} \label{Appendix: Modeling of the Meeting Scheduling Problem}

\subsection{Problem Formulation} \label{ssec:prob}

Let $\mathcal{E} = \{E_1,\dots,E_n\}$ denote the set of events we want to schedule and $\mathcal{P} = \{P_1,\dots,P_m\}$ the set of participants. Additionally, we define a function mapping each event to the set of its participants $\participants: \mathcal{E} \to 2^\mathcal{P}$, where $2^\mathcal{P}$ denotes the power set of $\mathcal{P}$. Let $days$ and $slots$ denote the number of days and time slots per day of our calendar (e.g., $days = 7, slots = 24$ would define a calendar for one week where each slot is 1 hour long). In order to add length to each event we define an additional function $\len : \mathcal{E} \to \N$, where $\N$ denotes the set of natural numbers (excluding 0). We do not limit the length; this allows for events to exceed a single day and even the entire calendar if needed. Finally, we assume that each participant has a preference for attending certain events at a given starting time, given by:
\begin{equation*}
	\pref : \mathcal{E} \times \participants(\mathcal{E}) \times \{1,\dots, days\} \times \{1,\dots, slots\} \to [0,1].
\end{equation*}
\noindent
For example, $\pref(E_1, P_1, 2, 6) = 0.7$ indicates that participant $P_1$ has a preference of 0.7 to attend event $E_1$ starting at day 2 and slot 6. The preference function allows participants to differentiate between different kinds of meetings (personal, business, etc.), or assign priorities. For example, one could be available in the evening for personal events while preferring to schedule business meetings in the morning.

Finding a schedule consists of finding a function that assigns each event to a given starting time, i.e.,
\begin{equation*}
	\sched : \mathcal{E} \to \left(\{1,\dots, days\} \times \{1,\dots, slots\}\right) \cup \emptyset 
\end{equation*}

\noindent
where $\sched(E) = \emptyset$ means that the event $E$ is not scheduled. For the schedule to be \emph{valid}, the following hard constraints need to be met:

\begin{enumerate}
	\item Scheduled events with common participants must not overlap.
	\item An event must not be scheduled at a (day, slot) tuple if any of the participants is not available. We represent an unavailable participant as one that has a preference of 0 (as given by the function $\pref$) for that event at the given (day, slot) tuple.
\end{enumerate}

More formally the hard constraints are:
\small
\begin{gather*}
	\forall E_1\in \mathcal{E}, \forall E_2 \in \mathcal{E}\setminus\{E_1\} :  \\
	(\sched(E_1) \neq \emptyset \land \sched(E_2) \neq \emptyset \land \participants(E_1) \cap \participants(E_2) \neq \emptyset) \\
	\implies (\sched(E_1) > \aju(E_2) \lor \sched(E_2) > \aju(E_1))
\end{gather*}
\noindent
and
\begin{gather*}
	\forall E\in \mathcal{E} : \\
	(\exists P \in \mathcal{P} , \exists d \in [1,days], \exists s \in [1,slots] : \pref(E,P,d,s) = 0 \\
	  \implies \sched(E) \neq (d,s))
\end{gather*}
\normalsize

\noindent
where $\aju(E)$ returns the ending time (last slot) of the event $E$ as calculated by the starting time $\sched(E)$ and the length $\len(E)$.

In addition to finding a valid schedule, we focus on maximizing the social welfare, i.e., the sum of the preferences for all scheduled meetings:

\begin{equation*}
	\sum_{\substack{E \in \mathcal{E}\\ \sched(E) \neq \emptyset}} \sum_{P \in \mathcal{P}}  \pref(E,P,\sched(E))
\end{equation*}

\subsection{Modeling Events, Participants, and Utilities}

\paragraph{Event Length} \label{ssec:gen}

To determine the length of each generated event, we used information on real meeting lengths in corporate America in the 80s \cite{romano01}. That data were then used to fit a logistic curve, which in turn was used to yield probabilities for an event having a given length. The function was designed such that the maximum number of hours was 11. According to the data less than $1\%$ of meetings exceeded that limit.

\begin{figure}[t!]
	\centering
	\begin{subfigure}[t]{0.48\linewidth}
		\centering
		\includegraphics[width=1\linewidth]{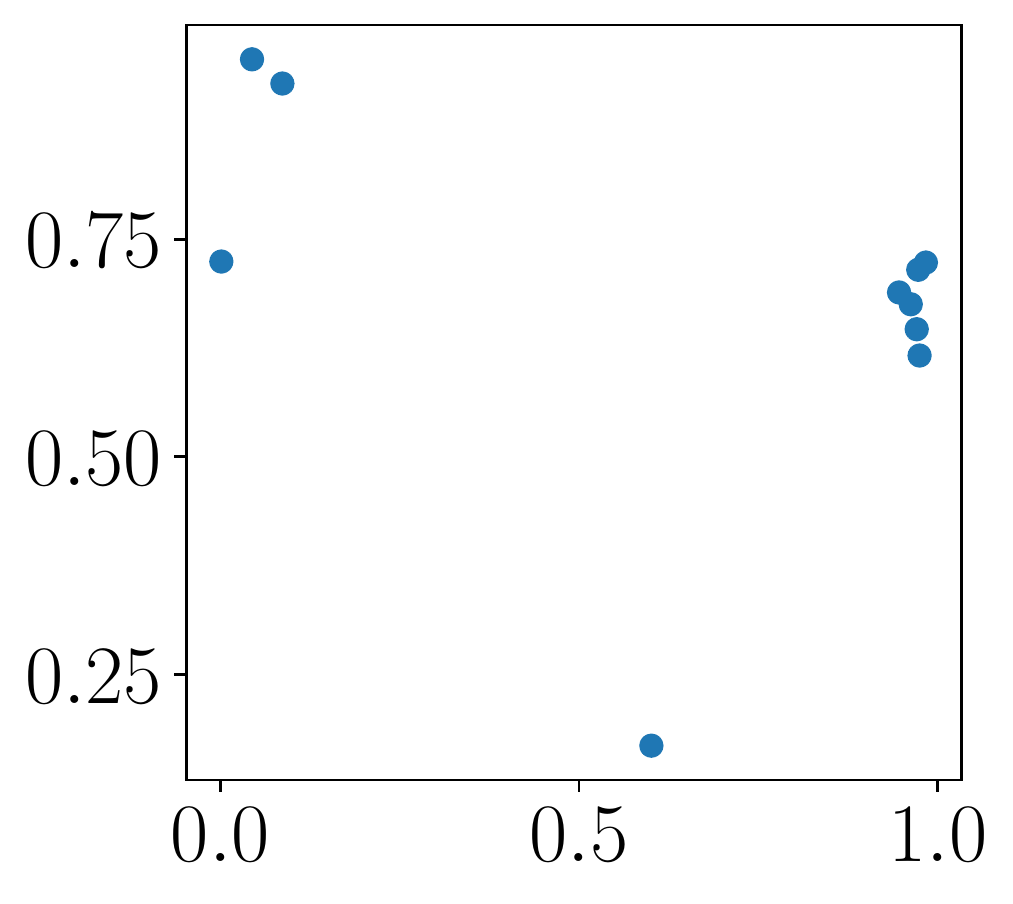}
		\caption{$p=10$}
	\end{subfigure}
	~ 
	\begin{subfigure}[t]{0.48\linewidth}
		\centering
		\includegraphics[width=1\linewidth]{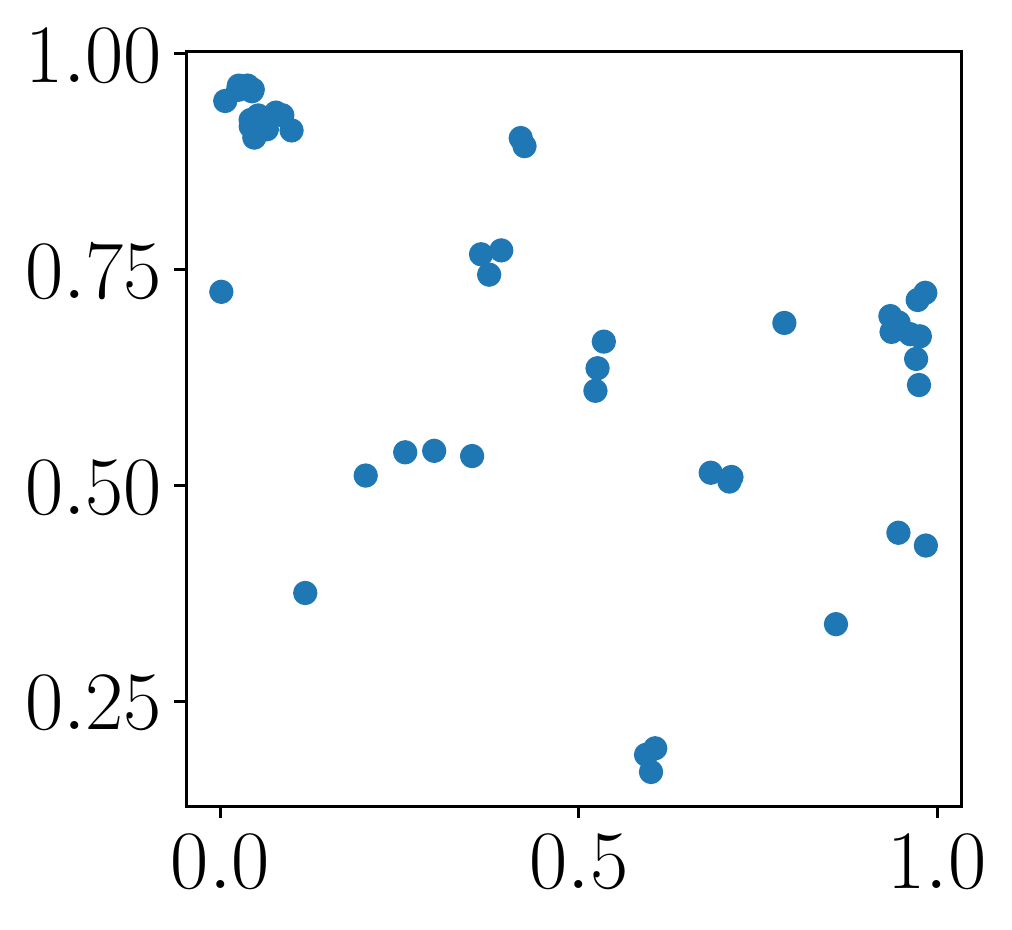}
		\caption{$p=50$}
	\end{subfigure}
	~ 
	\begin{subfigure}[t]{0.48\linewidth}
		\centering
		\includegraphics[width=1\linewidth]{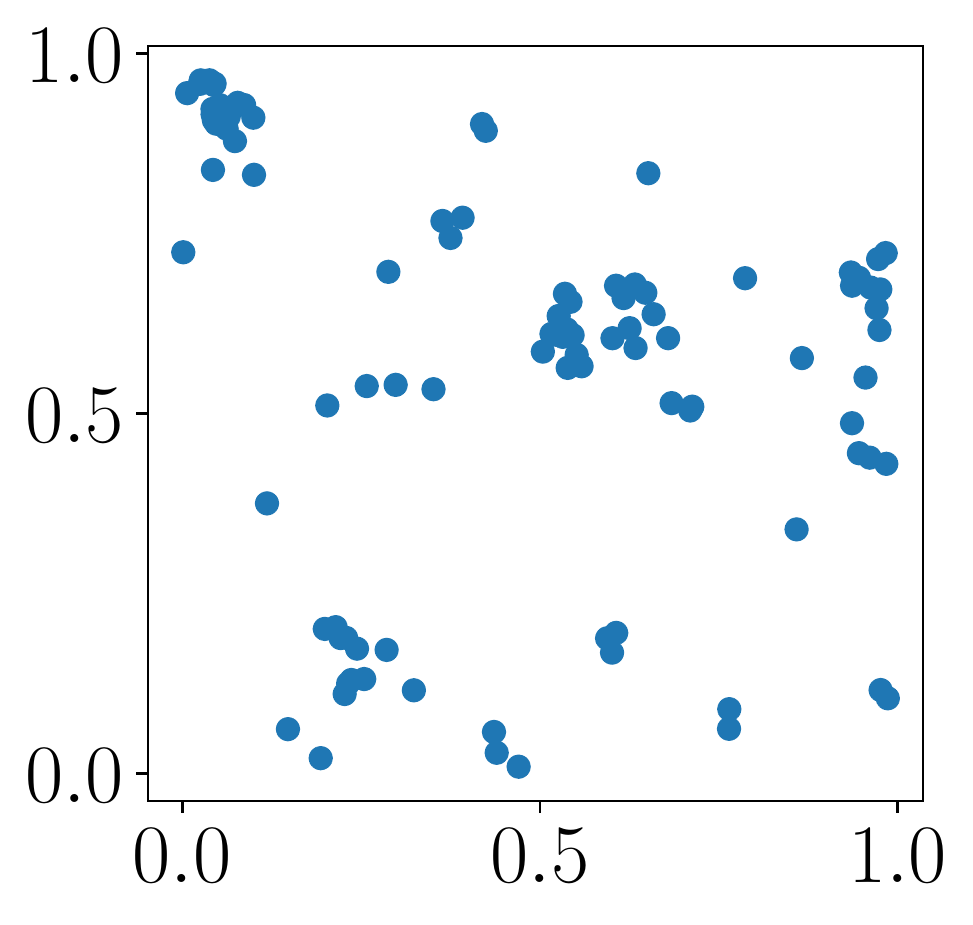}
		\caption{$p=100$}
	\end{subfigure}
	~ 
	\begin{subfigure}[t]{0.48\linewidth}
		\centering
		\includegraphics[width=1\linewidth]{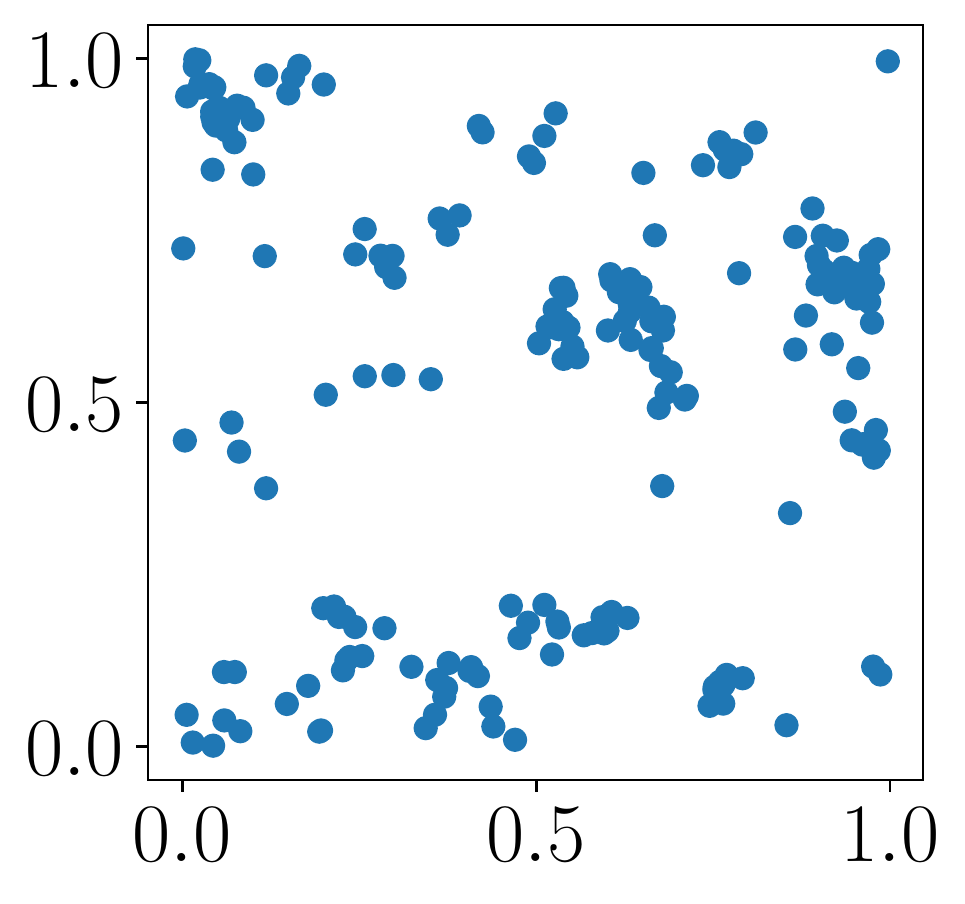}
		\caption{$p=200$}
	\end{subfigure}
	\caption{Generated datapoints on the $1\times 1$ plane for $p$ number of people.}
	\label{fig: appendixcluster}
\end{figure}

\paragraph{Participants} \label{ssec:att}

To determine the number of participants in an event, we used information on real meeting sizes \cite{romano01}. As before, we used the data to fit a logistic curve, which we used to sample the number of participants. The curve was designed such that no more than 90 people were chosen for an event at a time\footnote{The median is significantly lower.} (only $3\%$ of the meetings exceeds this number). Finally, for instances where the number of people in the entire system was below 90, that bound was reduced accordingly.

In order to simulate the naturally occurring clustering of people (see also \ref{ALMA: Computational Communication Complexity}) we assigned each participant to a point on a $1 \times 1$ plane in a way that enabled the emergence of clusters. Participants that are closer on the plane, are more likely to attend the same meeting. In more detail, we generated the points in an iterative way. The first participant was assigned a uniformly random point on the plane. For each subsequent participant, there is a $30\%$ probability that they also get assigned a uniformly random point. With a $70\%$ probability the person would be assigned a point based on a normal distribution centered at one of the previously created points. The selection of the latter point is based on the time of creation; a recently created point is exponentially more likely to be chosen as the center than an older one. This ensures the creation of clusters, while the randomness and the preference on recently generated points prohibits a single cluster to grow excessively. Figure \ref{fig: appendixcluster} displays an example of the aforedescribed process.

\paragraph{Utilities}
 
The utility function has two independent components. The first was designed to roughly reflect availability on an average workday. This function depends only on the time and is independent of the day. For example, a participant might prefer to schedule meetings in the morning rather than the afternoon (or during lunch time). A second function decreases the utility for an event over time. This function is independent of the time slot and only depends on the day and reflects the desire to schedule meetings sooner rather than days or weeks into the future. We generate the utility value for each participant / event combination using a normal distribution with the product of the aforementioned preference functions as mean and 0.1 as standard deviation. Figure \ref{img:dat} displays the distributions.

In addition to the above, we set a number of slots to 0 for each participant to simulate already scheduled meetings or other obligations. A maximum of $b$ slots were blocked this way each day. We chose which blocks to block using the same preference function for a day described above (Figure \ref{img:dat}, left). In other words, a slot with a high utility is also more likely to get blocked. Finally, for each event, we only consider the 24 most valuable slots.

\begin{figure}[t!]
	\centering
	\includegraphics[width=1\linewidth]{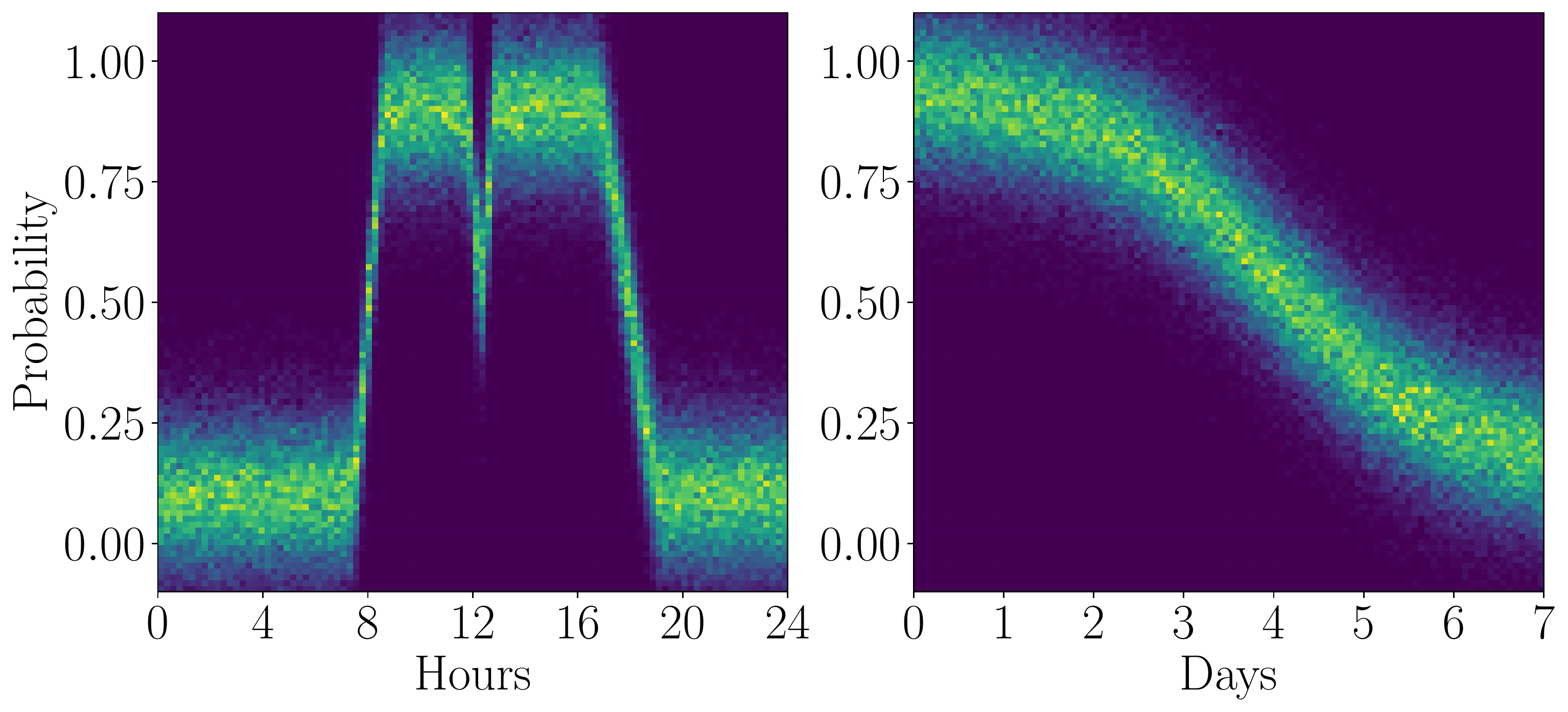}
	\caption{Distribution used for generating preference data. The image in the left displays the distribution of preference throughout a workday. The image on the right displays the distribution of preference throughout a week.}
	\label{img:dat}
\end{figure}

\subsection{Mapping the Meeting Scheduling Problem to the Allocation Problem} \label{Mapping the Meeting Scheduling Problem to the Allocation Problem}

We can map the meeting scheduling problem to the assignment problem of Section \ref{The Assignment Problem} if we consider each event as an agent and each starting time, i.e., (day,slot) tuple, as a resource. The utility for scheduling an event to a slot can be considered as the sum of the utilities (preferences) of the participants. There are, however, some important differences. For one, two events can be scheduled at the same time as long as their sets of respective participants are disjoint. On the other hand, two events can be conflicting even if they are not assigned to the same starting slot, if they overlap. However, if we detect these conflicts, we can apply ALMA (and ALMA-Learning) to solve the meeting scheduling problem. We describe how to do so in Section \ref{Implementation Details}.

\subsection{Implementation Details} \label{Implementation Details}

\paragraph{Event-agents and Representation-agents}

Each participant and each event is represented by an agent. We call those agents \emph{representation-agents} and \emph{event-agents}, respectively. For simplicity we use the variables introduced in \ref{ssec:prob} to describe these agents. I.e., $\mathcal{E}$ is the set of agents representing events and $\mathcal{P}$ the set of agents representing participants. For simplicity, we hereafter refer to a ($day$,$slot$) tuple simply as a slot. Initially each agent in $\mathcal{P}$ knows the events he wants to attend, their length, as well as the corresponding personal preferences. Additionally, each agent in $\mathcal{P}$ creates an empty personal calendar for events to be added later. The agents in $\mathcal{E}$ on the other hand know their set of participants.

\paragraph{Aggregation of Preferences}

Each event-agent in $\mathcal{E}$ needs to be able to aggregate the utilities of their attendees and compute the joint utility for the event. To do this, the representation-agents attending an event can simply relay their utilities to the corresponding event agents\footnote{Alternatively one could increase privacy by adding noise to the utilities or even by submitting a ranking of possible slots instead of actual utilities. The latter has the added benefit that the corresponding event-agent can choose a scale and does not rely on all users having the same metric when it comes to their utilities.}. The agents in $\mathcal{E}$ will then combine the received data (e.g., sum the utilities for each slot). Since we defined slots with a preference of 0 as unavailable, slots where one of the preferences is 0 will assume the value 0 instead of the sum. The slot entries will then be converted into a list $\mathcal{L} = [L_1,\dots,L_{days\cdot slots}]$, sorted in descending order, where each entry is a tuple of the form $L_i = \langle day_i,slot_i,preference_i \rangle$. Entries with a preference of 0 can simply be dropped.

\paragraph{ALMA: Collision Detection}

With the aforedescribed setup we can now apply ALMA. Each event-agent has a list of resources (i.e., slots) and the corresponding utility values. Since collisions are based on the attendees of each event, we let the representation-agents detect possible collisions. The event-agents send a message to all their attendees informing them of the currently contested slot. In turn, each representation-agent will check for possible collisions and relay that information back to the corresponding event-agents. Possible collisions are checked against all simultaneous attempts to acquire slots and an internal calendar for each representation-agent that contains already scheduled slots. This leaves the decision of collision detection to the representation-agents, which does not only make sure that the relevant information for scheduling a meeting stays within the group of participants, but also allows to enforce additional personal constraints, which were not considered in the problem description in Section \ref{ssec:prob}. For example, if two events are impossible to attend in succession due to locational constraints, the attendee could simply appear unavailable by reporting a collision. If there are no collisions, the event-agent will `acquire' the slot and inform his attendees. All informed representation-agents will delete the event-agent from their list of working agents. 

\paragraph{ALMA: Normalizing the Utilities}

The utilities for an event-agent are not in the range $[0,1]$, as required by ALMA; instead they are bounded by the number of attendees, as a result of summing up the individual utilities. Therefore, the loss is also not in $[0,1]$. There are several approaches for normalizing the utilities. One option is to simply divide by the number of attendees. This has the important downside that it distorts `the comparability of loss' between event-agents, since events with fewer attendees would be considered as important as ones with more participants even though the latter usually contributes more to social welfare. In other words, we want event-agents with more participants to be less likely to back-off in case of a conflict. Another option is to find a global variable to use for normalization such as the maximum number of attendees over all events or the highest utility value for any event. We chose the latter option.

\paragraph{ALMA: Computational and Communication Complexity} \label{ALMA: Computational Communication Complexity}

Let $R^*$ denote the maximum number of resources any agent is interested in, and $N^*$ denote the maximum number of agents competing for any resource. \citeauthor{ijcai201931} prove that by bounding these two quantities (i.e., we consider $R^*$ and $N^*$ to be constant functions of $N$, $R$), the expected number of steps any individual agent running ALMA requires to converge is independent of the total problem size (i.e., $N$, and $R$). For the meeting scheduling problem:

\begin{itemize}
	\item $R^*$: This corresponds to the maximum number of possible slots for any event. That number is bounded by the number of all available slots, given the calendar. However, in practice, meetings must be scheduled within a specific time-frame, thus each event would have a small, bounded number of slots available.
	\item $N^*$: This corresponds to the maximum number of event-agents with overlapping attendee sets competing for a slot. Assuming that we can cluster the total set of participants into smaller groups, such that either inter-group events are rare or preferred meeting times of groups do not overlap, then we could describe $N^*$ as the maximal number of events of any such group. This form of clustering naturally occurs in many real-life situations, e.g., people in the corporate world often build clusters by their work hours or their respective departments. 
\end{itemize}

To summarize, the number of rounds is bounded by the number of possible slots per event and the number of events per cluster, if a clustering exists.

In addition to the number of agents and resources, the complexity bound proven by \citeauthor{ijcai201931} depends on the worst-case back-off probability, $p^*$. To mitigate this problem and further speed-up the run-time in real-world applications, we can slowly reduce the steepness of the back-off curve over time.

Finally, each round requires $\Landau(E^* + A^*)$ messages, where $E^*$ is the maximum number of events for any single attendee and $A^*$ the maximum number of attendees for any single event.

\paragraph{MSRAC baseline}

The problem formulation of \citeauthor{benhassine07} requires each meeting to be associated with a `degree of importance'. To calculate this value, we take the average utility of all participants and multiply it with the number of participants\footnote{Since meetings with more participants contribute more to social welfare.}. This makes it highly unlikely that two events have the same importance value, thus reducing collisions of equally important events.

We modified the MSRAC algorithm \cite{benhassine07} accordingly, to account for events of varying length and preferences per event/slot combinations. In short, MSRAC works as follows. The event-agent asks all participants for the relevant utility information. These are then summed up and sorted. Unavailable slots, as well as slots occupied by more important events, are dropped. The event-agent proposes the first slot in the sorted list. Then, each participant considers all the proposed events for that slot and keeps the one with the highest importance value, specifically:

\begin{itemize}
	\item If the slot is free, accept the proposal.
	\item If the slot is occupied by a less important event, accept the proposal and invite the agent of the less important event to reschedule.
	\item If the slot is occupied by a more important event, reject proposal and invite agent to propose new slot.
	\item If the slot is occupied by an equally important event, keep the one with a higher utility and reject and reschedule the other one. 
\end{itemize}

If an event-agent receives a message to reschedule, it deletes the current proposal and proposes the next slot to its participants. This process is repeated until a stable state is reached.

\section{Evaluation} \label{Appendix: Evaluation}

\paragraph{Fairness Metrics}

Given the broad literature on fairness, we measured two different fairness indices:

\textbf{(a)}  \emph{The Jain index} \cite{DBLP:journals/corr/cs-NI-9809099}: Widely used in network engineering to determine whether users or applications receive a fair share of system resources. It exhibits a lot of desirable properties such as: population size independence, continuity, scale and metric independence, and boundedness. For an allocation of $N$ agents, such that the $n^{\text{th}}$ agent is alloted $x_n$, the Jain index is given by \autoref{eq: Jain Index}. An allocation $\mathbf{x} = (x_1, \dots, x_N) ^\top$ is considered fair, iff $\mathds{J}(\mathbf{x}) = 1$.

\small
\begin{equation} \label{eq: Jain Index}
	\mathds{J}(\mathbf{x}) = \frac{\left(\underset{n = 1}{\overset{N}{\sum}} x_n\right) ^ 2}{N \underset{n = 1}{\overset{N}{\sum}}  x_n ^ 2}
\end{equation}
\normalsize

\textbf{(b)} \emph{The Gini coefficient} \cite{gini1912variabilita}: One of the most commonly used measures of inequality by economists intended to represent the wealth distribution of a population of a nation. For an allocation game of $N$ agents, such that the $n^{\text{th}}$ agent is alloted $x_n$, the Gini coefficient is given by \autoref{eq: Gini coefficient}. A Gini coefficient of zero expresses perfect equality, i.e., an allocation is fair iff $\mathds{G}(\mathbf{x}) = 0$.

\small
\begin{equation} \label{eq: Gini coefficient}
	\mathds{G}(\mathbf{x}) = \frac{\underset{n = 1}{\overset{N}{\sum}} \underset{n' = 1}{\overset{N}{\sum}} \left| x_n - x_{n'} \right|}{2 N \underset{n = 1}{\overset{N}{\sum}}  x_n}
\end{equation}
\normalsize

\subsection{Test Case \#1: Synthetic Benchmarks} \label{Evaluation: Synthetic Benchmarks}

We run each configuration 16 times and report the average values. Since we have randomized algorithms, we also run each problem instance of each configuration 16 times. ALMA, and ALMA-Learning's parameters were set to: $\alpha = 0.1, \beta = 2, \epsilon = 0.01, L = 20$.

\paragraph{Social Welfare}

We begin with the loss in social welfare compared to the optimal solution. Figures \ref{fig: appendix map_tt_512_RDSW}, \ref{fig: appendix noisyCommonUtilities_tt_8192_RDSW}, \ref{fig: appendix noisyCommonUtilities_tt_8192_un_02_RDSW}, \ref{fig: appendix noisyCommonUtilities_tt_8192_un_04_RDSW}, and \ref{fig: appendix binary_tt_512_RDSW} present the results for the three test-cases. Table \ref{tb: supp social welfare} aggregates the results.

In all of the test-cases, ALMA-Leaning is able to quickly reach \emph{near-optimal} allocations. On par with previous results \cite{ijcai201931,danassis2019putting}, ALMA loses around $10\%$ to $15\%$\footnote{Note that both ALMA and ALMA-Learning use the same function $P(loss) = f(loss)^\beta$ (see Equation \ref{Eq: loss}) to compute the back-off probability, in order to provide a fair common ground for the evaluation.}. Is is worth noting that test-cases that are `harder' for ALMA -- specifically test-case (a) Map, where ALMA maintains the same gap on the optimal solution as the number of resources grow, and test-case (c) Binary Utilities, where ALMA exhibits the highest loss for 16 resources\footnote{Binary utilities represent a somewhat adversarial test-case for ALMA, since the agents can not utilize the more sophisticated back-off mechanism based on the loss (loss is either 1, or 0 in this case).} -- are `easier' to learn for ALMA-Learning. In the aforementioned two test cases, ALMA-Learning was able to learn near-optimal to optimal allocations in just $64-512$ training time-steps (in fact in certain cases it learns near-optimal allocations in as little as 32 time-steps). Contrary to that, in test-case (b) Noisy Common Utilities, ALMA-Learning requires significantly more time to learn (we trained for 8192 time-steps), especially for larger games ($R > 256$). Intuitively we believe that this is because ALMA already starts with a near optimal allocation, and given the high similarity on the agent's utility tables (especially for $\sigma = 0.1$), it requires a lot of fine-tunning to improve the result.

\paragraph{Fairness} \label{Fairness}

Next, we evaluate the fairness of the final allocation. Figures \ref{fig: appendix map_tt_512_JI}, \ref{fig: appendix noisyCommonUtilities_tt_8192_JI}, \ref{fig: appendix noisyCommonUtilities_tt_8192_un_02_JI}, \ref{fig: appendix noisyCommonUtilities_tt_8192_un_04_JI}, and \ref{fig: appendix binary_tt_512_JI} depict the Jain index (higher is better) for the three test-cases, while the Gini coefficient (lower is better) is presented in Figures \ref{fig: appendix map_tt_512_GI}, \ref{fig: appendix noisyCommonUtilities_tt_8192_GI}, \ref{fig: appendix noisyCommonUtilities_tt_8192_un_02_GI}, \ref{fig: appendix noisyCommonUtilities_tt_8192_un_04_GI}, and \ref{fig: appendix binary_tt_512_GI}. Tables \ref{tb: supp fairness jain} and \ref{tb: supp fairness gini} aggregate the results for both metrics. In all of the test-cases and for both indices, ALMA-Leaning achieves the most fair allocations, fairer than the optimal (in terms of social welfare) solution\footnote{To improve readability, we have omitted the results for test-case (b) Noisy Common Utilities for $\sigma = 0.2$ and $\sigma = 0.4$. In both cases ALMA-Learning performed better than the reported results for $\sigma = 0.1$.}.

\subsection{Test Case \#2: Meeting Scheduling}

We run each configuration 10 times and report the average values. The time limit for the CPLEX optimizer was set to 20 minutes. CPLEX ran on an Intel i7-7700HQ CPU (4 cores, 2.8 Ghz base clock) and 16 GB of ram. 

We used the SMAC\footnote{SMAC (sequential model-based algorithm configuration) is an automated algorithm configuration tool that uses Bayesian optimization (\url{https://www.automl.org/automated-algorithm-design/algorithm-configuration/smac/}).} tool \cite{hutter11} to choose the hyper-parameters. This resulted in selecting a logistic curve to compute the back-off probability, specifically Equation \ref{Eq: logistic}, where $\gamma = 15.72$. To compute the loss, we used Equation \ref{Eq: Meeting Scheduling loss}, where $k = 13$. ALMA-Learning's parameters were set to: $\alpha = 0.1, L = 20$.

\small
\begin{equation} \label{Eq: Meeting Scheduling loss}
	loss_n^i = \frac{\underset{j = i + 1}{\overset{k}{\sum}} \left( u_n(r_i) - u_n(r_j) \right)}{k - i}
\end{equation}
\normalsize

\small
\begin{equation} \label{Eq: logistic}
	f(loss) = \frac{1}{1 + e^{-\gamma(0.5 - loss)}}
\end{equation}
\normalsize

Figures \ref{fig: appendixmeeting sw}, \ref{fig: appendix meeting fairness jain}, and \ref{fig: appendix meeting fairness gini} present the results for the relative difference in social welfare (compared to CPLEX), and the Jain index and Gini coefficient, respectively. Moreover, in Figure \ref{fig: appendix meeting large} we depict the metrics for the hand-crafted larger instance. Finally, Tables \ref{tb: supp meeting scheduling social welfare}, \ref{tb: supp meeting scheduling jain}, and \ref{tb: supp meeting scheduling gini} aggregate the results.

\clearpage

\begin{table*}[t!]
	\centering
	\caption{Range of the average loss ($\%$) in social welfare compared the (centralized) optimal for the three different benchmarks.}
	\label{tb: supp social welfare}
	\resizebox{\textwidth}{!}{%
	\begin{tabular}{@{}lrcc@{}}
	\toprule
	\multicolumn{1}{c}{}                       & Greedy               & ALMA                 & \textbf{ALMA-Learning}                       \\ \midrule
	(a) Map                                    & $1.51\% - 18.71\%$   & $0.00\% - 9.57\%$    & $0.00\% - 0.89\%$ (512 training time-steps)  \\
	                                           & \multicolumn{1}{l}{} & \multicolumn{1}{l}{} & $0.00\% - 1.68\%$ (64 training time-steps)   \\
	(b) Noisy Common Utilities, $\sigma = 0.1$ & $8.13\% - 12.86\%$   & $2.96\% - 10.58\%$   & $1.34\% - 2.26\%$ (8192 training time-steps) \\
	(b) Noisy Common Utilities, $\sigma = 0.2$ & $5.40\% - 14.11\%$   & $1.37\% - 12.33\%$   & $0.35\% - 1.97\%$ (8192 training time-steps) \\
	(b) Noisy Common Utilities, $\sigma = 0.4$ & $0.79\% - 12.64\%$   & $0.75\% - 10.74\%$   & $0.05\% - 2.26\%$ (8192 training time-steps) \\
	(c) Binary Utilities                       & $0.10\% - 14.70\%$   & $0.00\% - 16.88\%$   & $0.00\% - 0.39\%$ (64 training time-steps)   \\ \bottomrule
	\end{tabular}%
	}
\end{table*}

\begin{table*}[t!]
	\centering
	\caption{Fairness -- Jain index (the higher, the better) for the three different benchmarks. In parenthesis we include the average improvement in fairness of ALMA-Learning compared to ALMA, Greedy, and Hungarian, respectively.}
	\label{tb: supp fairness jain}
	\resizebox{\textwidth}{!}{%
	\begin{tabular}{@{}llrcc@{}}
	\toprule
	\multicolumn{1}{c}{}                       & Hungarian     & Greedy        & ALMA          & \textbf{ALMA-Learning}                         \\ \midrule
	(a) Map                                    & $0.79 - 0.86$ & $0.70 - 0.73$ & $0.75 - 0.80$ & $0.86 - 0.89$ ($11.95\%$, $22.44\%$, $5.03\%$) \\
	(b) Noisy Common Utilities, $\sigma = 0.1$ & $0.81 - 0.92$ & $0.77 - 0.88$ & $0.79 - 0.89$ & $0.85 - 0.93$ ($5.58\%$, $7.58\%$, $1.81\%$)   \\
	(c) Binary Utilities                       & $0.84 - 1.00$ & $0.72 - 1.00$ & $0.82 - 0.98$ & $0.88 - 1.00$ ($10.18\%$, $5.36\%$, $0.58\%$)  \\ \bottomrule
	\end{tabular}%
	}
\end{table*}

\begin{table*}[t!]
	\centering
	\caption{Fairness -- Gini Coefficient (the lower, the better) for the three different benchmarks. In parenthesis we include the average improvement (decrease in inequality) of ALMA-Learning compared to ALMA, Greedy, and Hungarian, respectively.}
	\label{tb: supp fairness gini}
	\resizebox{\textwidth}{!}{%
	\begin{tabular}{@{}llrcc@{}}
	\toprule
	\multicolumn{1}{c}{}                       & Hungarian     & Greedy        & ALMA          & \textbf{ALMA-Learning}                            \\ \midrule
	(a) Map                                    & $0.19 - 0.23$ & $0.30 - 0.34$ & $0.25 - 0.28$ & $0.17 - 0.20$ ($-29.04\%$, $-42.91\%$, $-9.63\%$) \\
	(b) Noisy Common Utilities, $\sigma = 0.1$ & $0.17 - 0.24$ & $0.21 - 0.29$ & $0.19 - 0.28$ & $0.16 - 0.23$ ($-18.29\%$, $-23.66\%$, $-6.52\%$) \\
	(c) Binary Utilities                       & $0.00 - 0.16$ & $0.00 - 0.28$ & $0.02 - 0.18$ & $0.00 - 0.13$ ($-90.43\%$, $-92.61\%$, $-0.18\%$) \\ \bottomrule
	\end{tabular}%
	}
\end{table*}

\begin{figure*}[t!]
	\centering
	\begin{subfigure}[t]{0.32\textwidth}
		\centering
		\includegraphics[width = 1 \linewidth, trim={0em 0em 0em 0em}, clip]{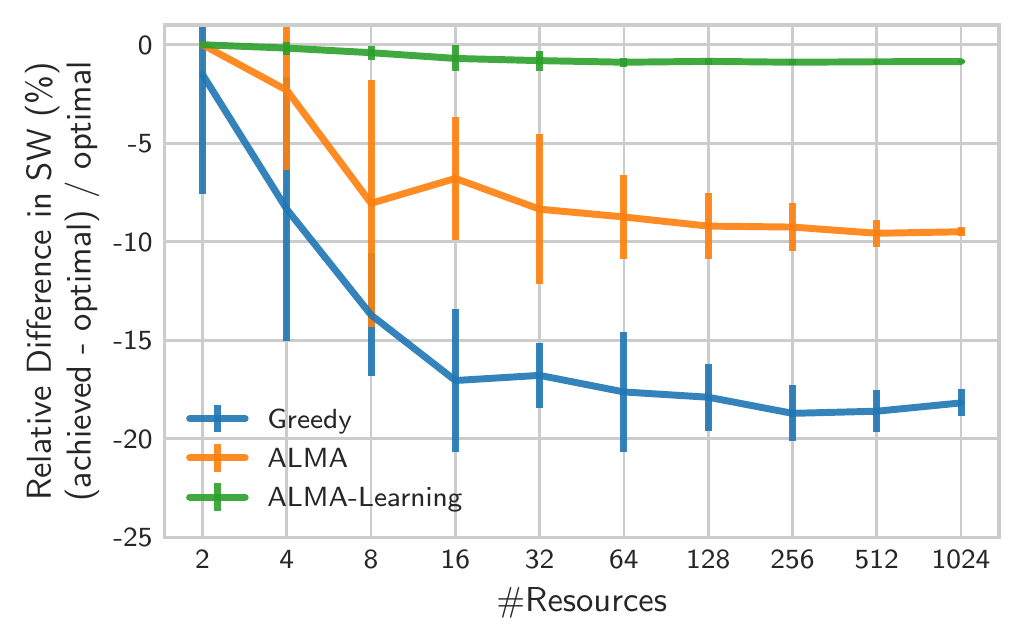}
		\caption{Relative Difference in Social Welfare (\%)}
		\label{fig: appendix map_tt_512_RDSW}
	\end{subfigure}
	~ 
	\begin{subfigure}[t]{0.32\textwidth}
		\centering
		\includegraphics[width = 1 \linewidth, trim={0em 0em 0em 0em}, clip]{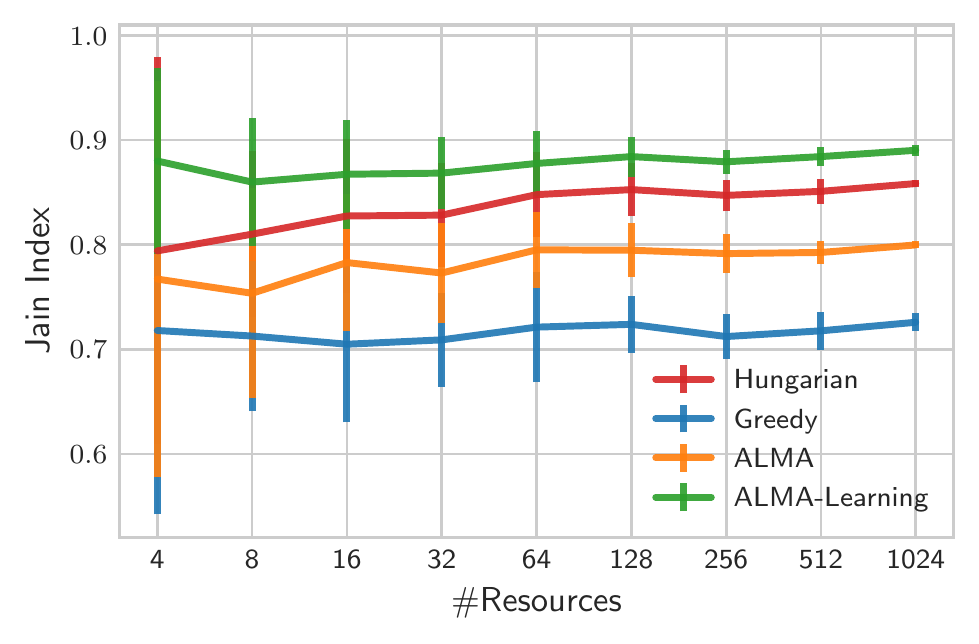}
		\caption{Jain Index (higher is better)}
		\label{fig: appendix map_tt_512_JI}
	\end{subfigure}
	~ 
	\begin{subfigure}[t]{0.32\textwidth}
		\centering
		\includegraphics[width = 1 \linewidth, trim={0em 0em 0em 0em}, clip]{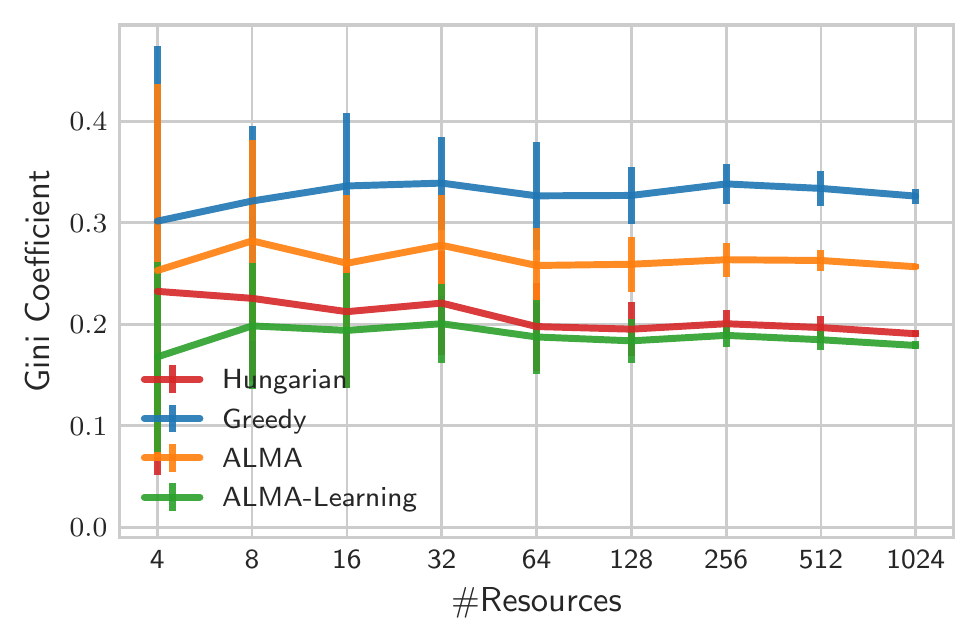}
		\caption{Gini Coefficient (lower is better)}
		\label{fig: appendix map_tt_512_GI}
	\end{subfigure}%
	\caption{Test-case (a): Map, we report the relative difference in social welfare, the Jain index, and the Gini coefficient, for increasing number of resources ($[2, 1024]$, $x$-axis in log scale), and $N = R$. For each problem instance, we trained ALMA-Learning for 512 time-steps.\medskip}
	\label{fig: appendix map_tt_512}
\end{figure*}

\begin{figure*}[t!]
	\centering
	\begin{subfigure}[t]{0.32\textwidth}
		\centering
		\includegraphics[width = 1 \linewidth, trim={0em 0em 0em 0em}, clip]{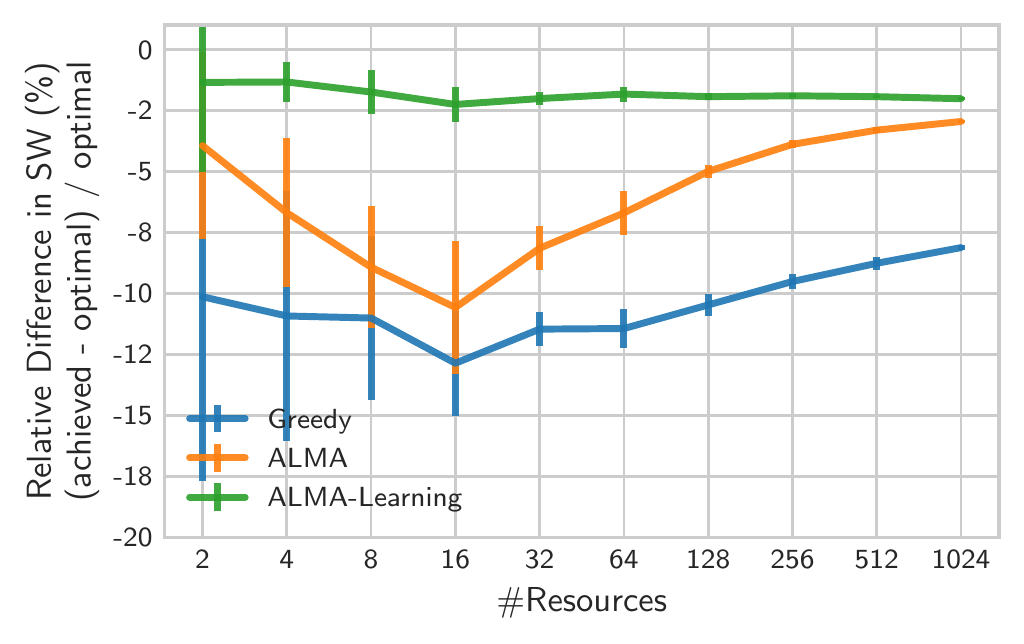}
		\caption{Relative Difference in Social Welfare (\%)}
		\label{fig: appendix noisyCommonUtilities_tt_8192_RDSW}
	\end{subfigure}
	~ 
	\begin{subfigure}[t]{0.32\textwidth}
		\centering
		\includegraphics[width = 1 \linewidth, trim={0em 0em 0em 0em}, clip]{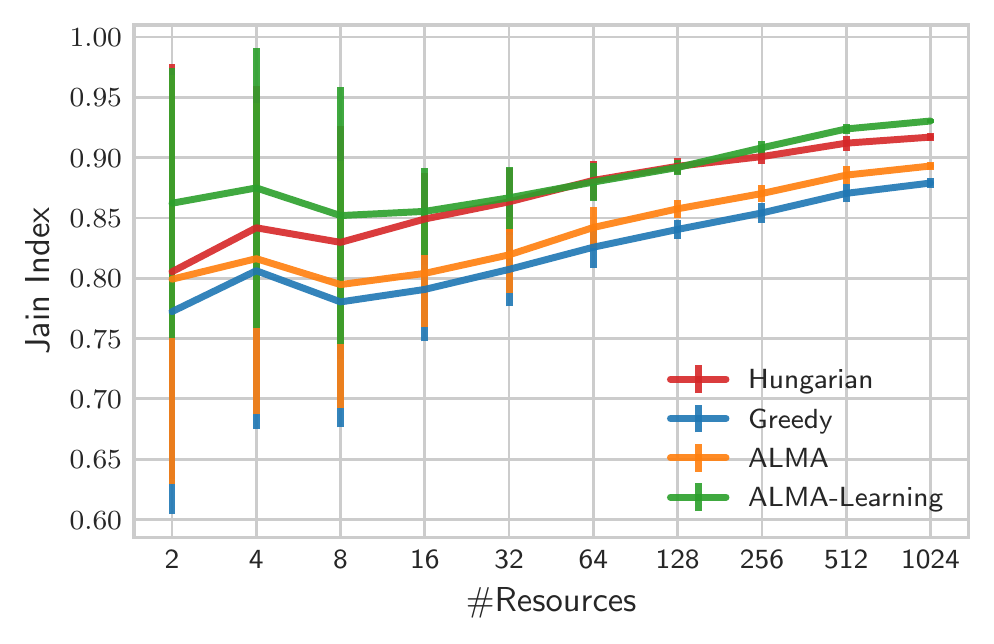}
		\caption{Jain Index (higher is better)}
		\label{fig: appendix noisyCommonUtilities_tt_8192_JI}
	\end{subfigure}
	~ 
	\begin{subfigure}[t]{0.32\textwidth}
		\centering
		\includegraphics[width = 1 \linewidth, trim={0em 0em 0em 0em}, clip]{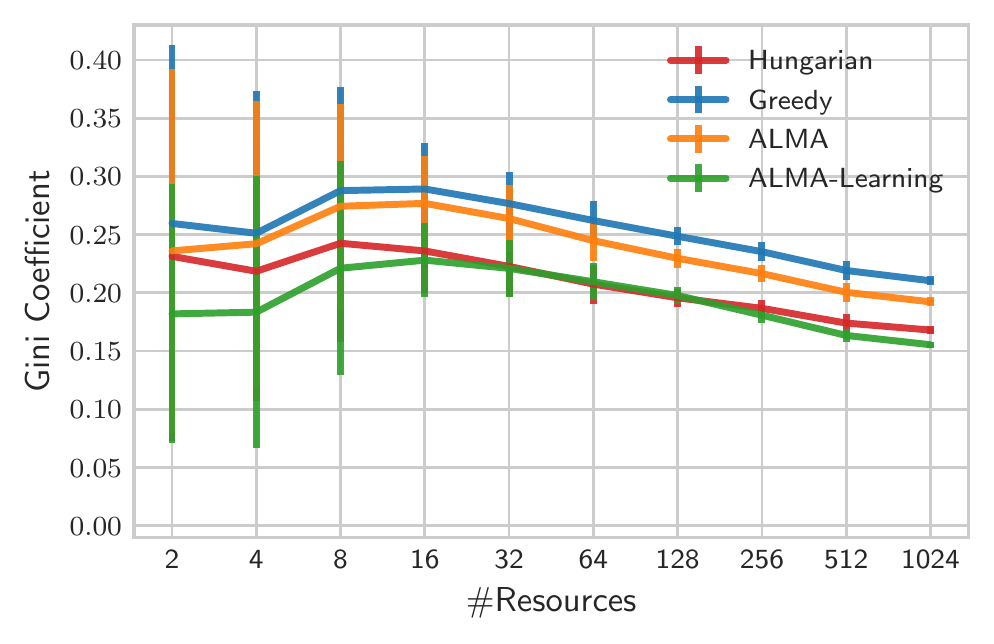}
		\caption{Gini Coefficient (lower is better)}
		\label{fig: appendix noisyCommonUtilities_tt_8192_GI}
	\end{subfigure}%
	\caption{Test-case (b): Noisy Common Utilities, $\sigma = 0.1$, we report the relative difference in social welfare, the Jain index, and the Gini coefficient, for increasing number of resources ($[2, 1024]$, $x$-axis in log scale), and $N = R$. For each problem instance, we trained ALMA-Learning for 8192 time-steps.\medskip}
	\label{fig: appendix noisyCommonUtilities_tt_8192}
\end{figure*}

\begin{figure*}[t!]
	\centering
	\begin{subfigure}[t]{0.32\textwidth}
		\centering
		\includegraphics[width = 1 \linewidth, trim={0em 0em 0em 0em}, clip]{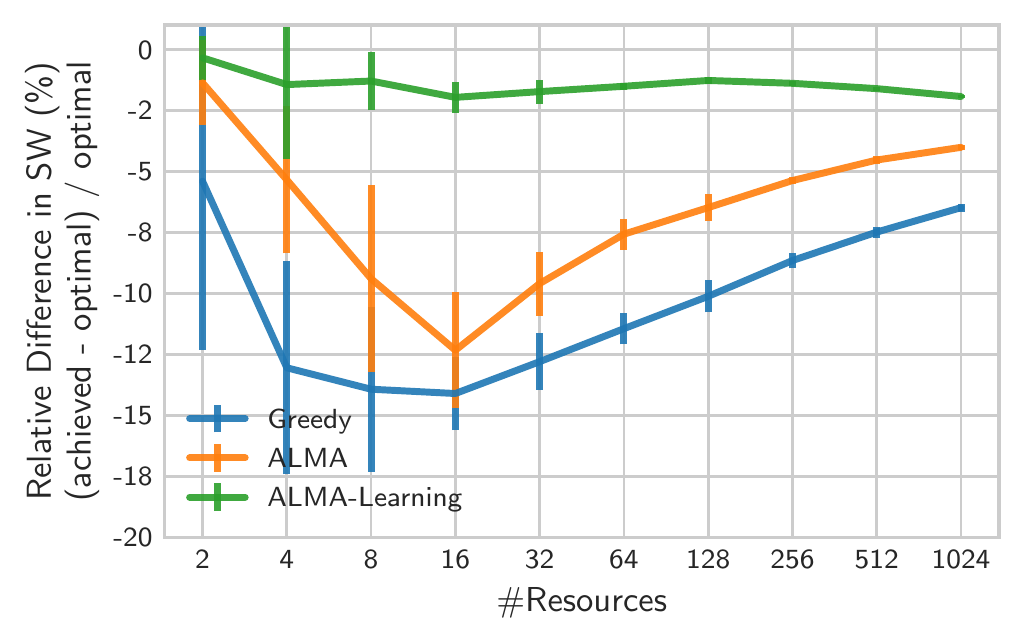}
		\caption{Relative Difference in Social Welfare (\%)}
		\label{fig: appendix noisyCommonUtilities_tt_8192_un_02_RDSW}
	\end{subfigure}
	~ 
	\begin{subfigure}[t]{0.32\textwidth}
		\centering
		\includegraphics[width = 1 \linewidth, trim={0em 0em 0em 0em}, clip]{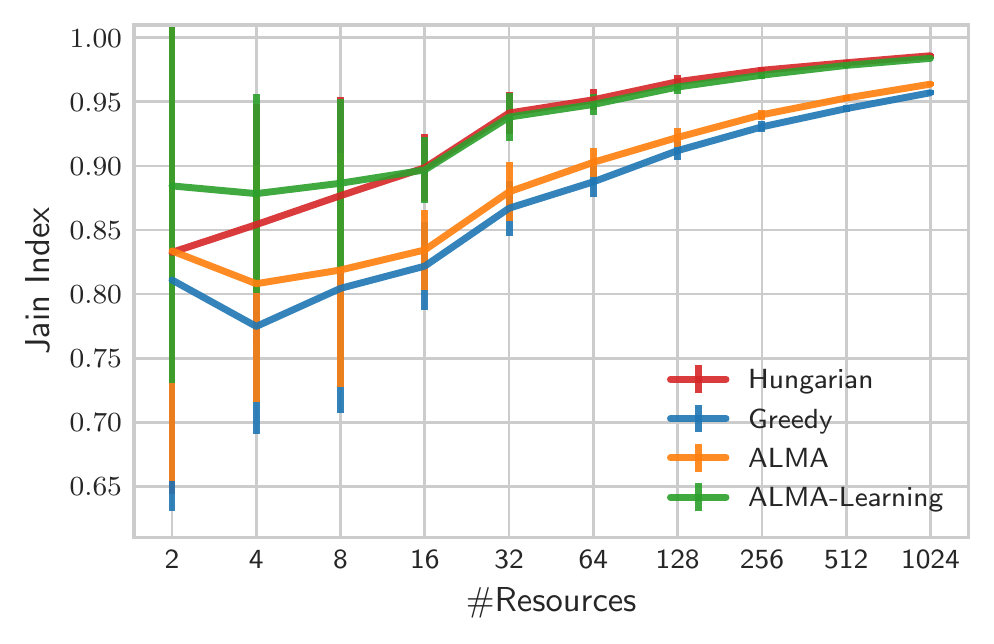}
		\caption{Jain Index (higher is better)}
		\label{fig: appendix noisyCommonUtilities_tt_8192_un_02_JI}
	\end{subfigure}
	~ 
	\begin{subfigure}[t]{0.32\textwidth}
		\centering
		\includegraphics[width = 1 \linewidth, trim={0em 0em 0em 0em}, clip]{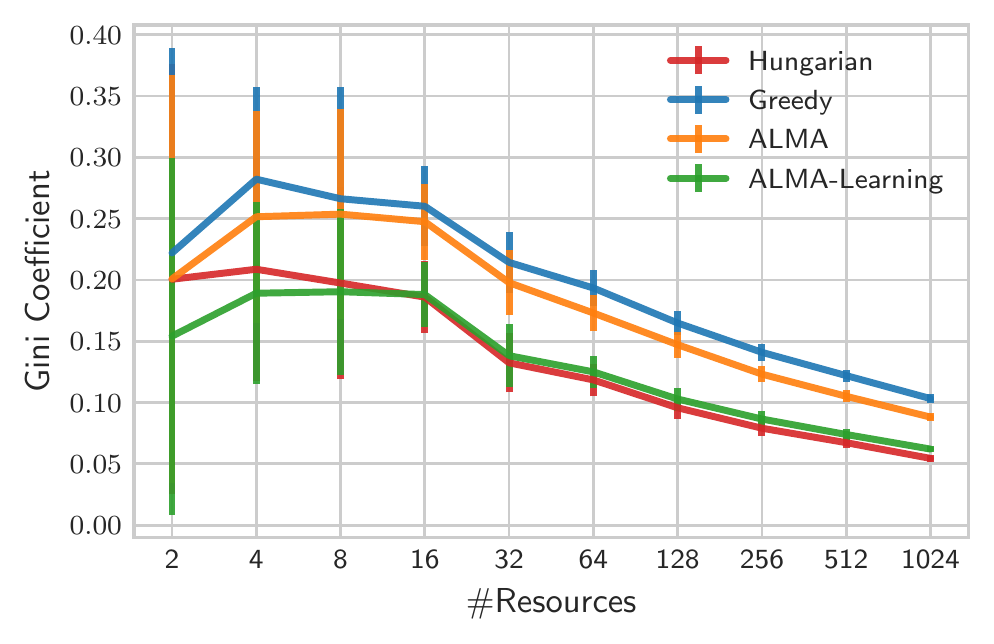}
		\caption{Gini Coefficient (lower is better)}
		\label{fig: appendix noisyCommonUtilities_tt_8192_un_02_GI}
	\end{subfigure}%
	\caption{Test-case (b): Noisy Common Utilities, $\sigma = 0.2$, we report the relative difference in social welfare, the Jain index, and the Gini coefficient, for increasing number of resources ($[2, 1024]$, $x$-axis in log scale), and $N = R$. For each problem instance, we trained ALMA-Learning for 8192 time-steps.\medskip}
	\label{fig: appendix noisyCommonUtilities_tt_8192_un_02}
\end{figure*}

\begin{figure*}[t!]
	\centering
	\begin{subfigure}[t]{0.32\textwidth}
		\centering
		\includegraphics[width = 1 \linewidth, trim={0em 0em 0em 0em}, clip]{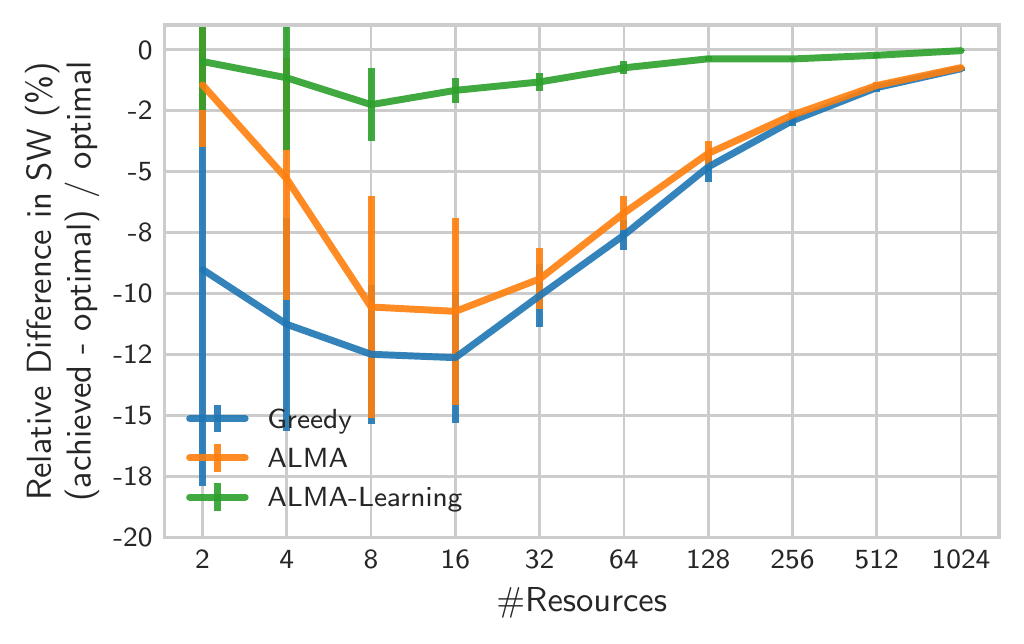}
		\caption{Relative Difference in Social Welfare (\%)}
		\label{fig: appendix noisyCommonUtilities_tt_8192_un_04_RDSW}
	\end{subfigure}
	~ 
	\begin{subfigure}[t]{0.32\textwidth}
		\centering
		\includegraphics[width = 1 \linewidth, trim={0em 0em 0em 0em}, clip]{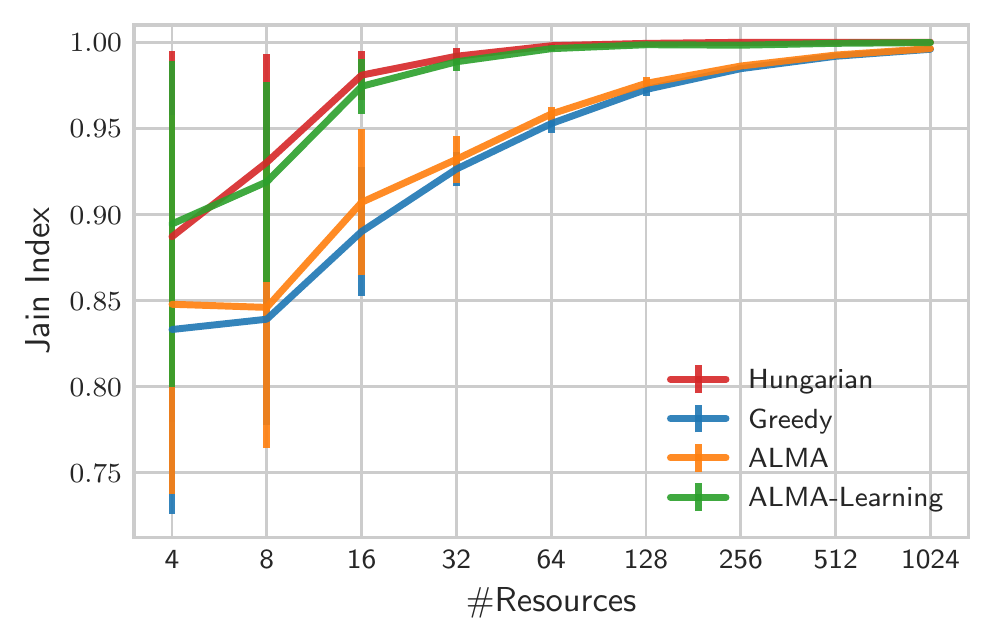}
		\caption{Jain Index (higher is better)}
		\label{fig: appendix noisyCommonUtilities_tt_8192_un_04_JI}
	\end{subfigure}
	~ 
	\begin{subfigure}[t]{0.32\textwidth}
		\centering
		\includegraphics[width = 1 \linewidth, trim={0em 0em 0em 0em}, clip]{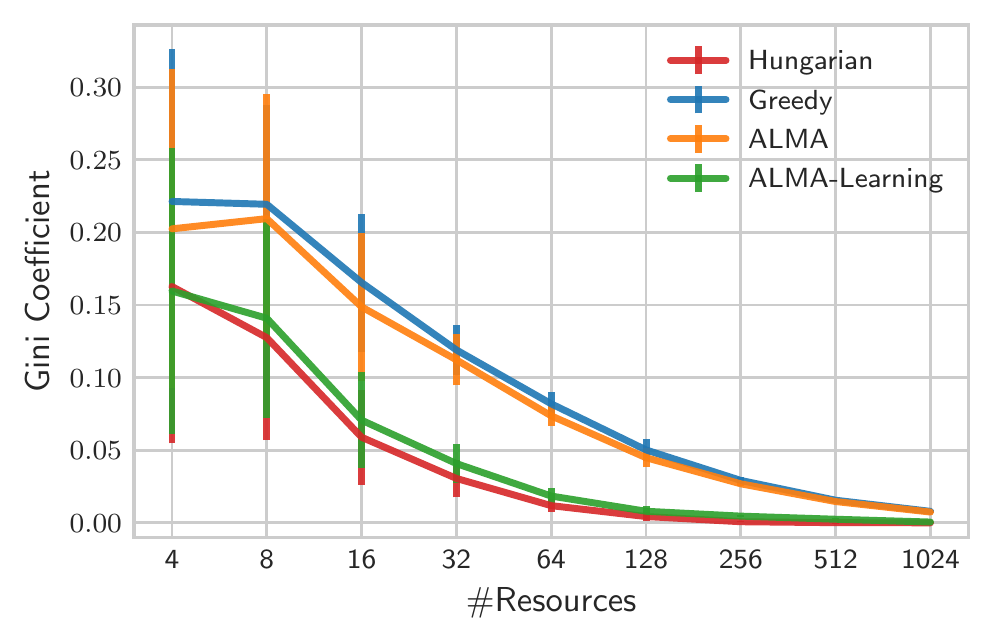}
		\caption{Gini Coefficient (lower is better)}
		\label{fig: appendix noisyCommonUtilities_tt_8192_un_04_GI}
	\end{subfigure}%
	\caption{Test-case (b): Noisy Common Utilities, $\sigma = 0.4$, we report the relative difference in social welfare, the Jain index, and the Gini coefficient, for increasing number of resources ($[2, 1024]$, $x$-axis in log scale), and $N = R$. For each problem instance, we trained ALMA-Learning for 8192 time-steps.\medskip}
	\label{fig: appendix noisyCommonUtilities_tt_8192_un_04}
\end{figure*}

\begin{figure*}[t!]
	\centering
	\begin{subfigure}[t]{0.32\textwidth}
		\centering
		\includegraphics[width = 1 \linewidth, trim={0em 0em 0em 0em}, clip]{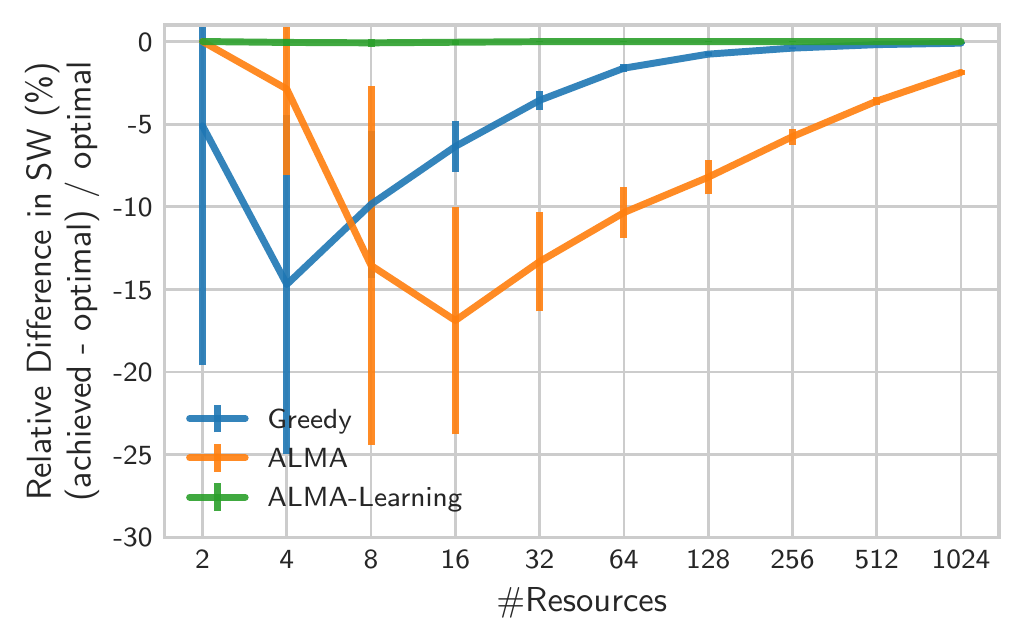}
		\caption{Relative Difference in Social Welfare (\%)}
		\label{fig: appendix binary_tt_512_RDSW}
	\end{subfigure}
	~ 
	\begin{subfigure}[t]{0.32\textwidth}
		\centering
		\includegraphics[width = 1 \linewidth, trim={0em 0em 0em 0em}, clip]{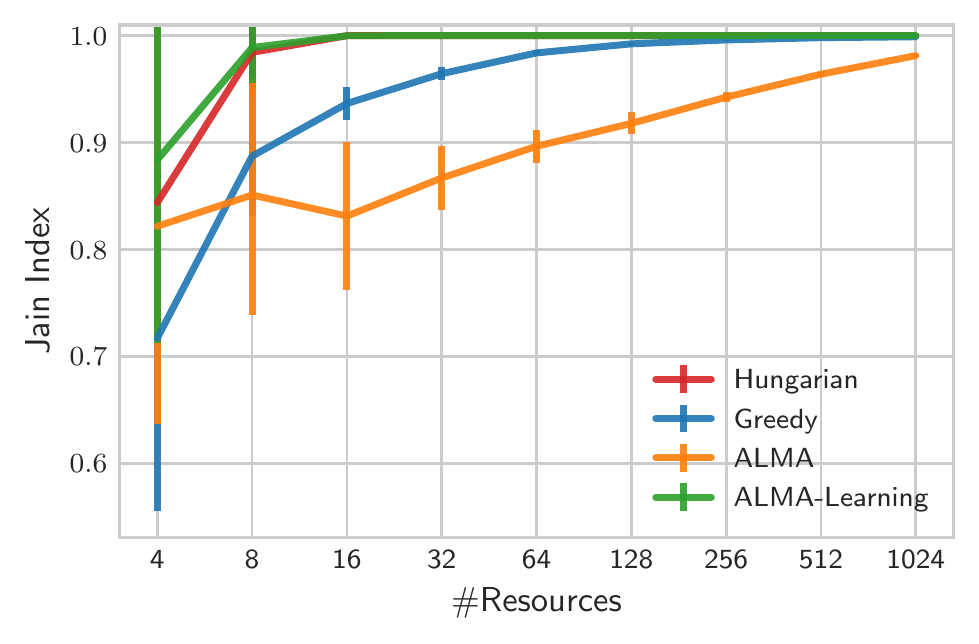}
		\caption{Jain Index (higher is better)}
		\label{fig: appendix binary_tt_512_JI}
	\end{subfigure}
	~ 
	\begin{subfigure}[t]{0.32\textwidth}
		\centering
		\includegraphics[width = 1 \linewidth, trim={0em 0em 0em 0em}, clip]{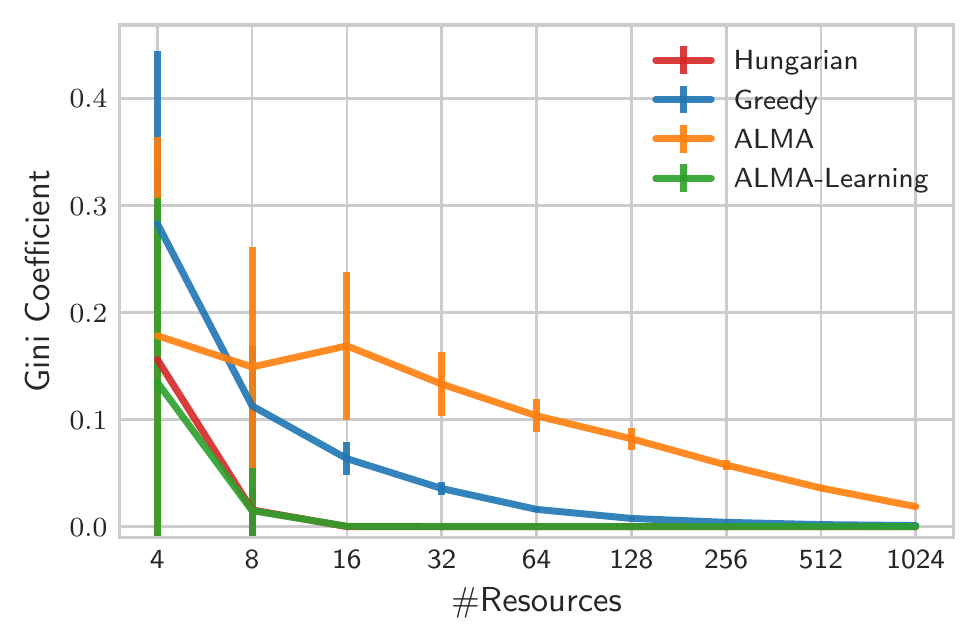}
		\caption{Gini Coefficient (lower is better)}
		\label{fig: appendix binary_tt_512_GI}
	\end{subfigure}%
	\caption{Test-case (c): Binary Utilities, we report the relative difference in social welfare, the Jain index, and the Gini coefficient, for increasing number of resources ($[2, 1024]$, $x$-axis in log scale), and $N = R$. For each problem instance, we trained ALMA-Learning for 64 time-steps.\medskip}
	\label{fig: appendix binary_tt_512}
\end{figure*}


\clearpage

\begin{figure*}[t!]
	\centering
	\begin{subfigure}[t]{0.32\textwidth}
		\centering
		\includegraphics[width = 1 \linewidth, trim={0em 0em 0em 0em}, clip]{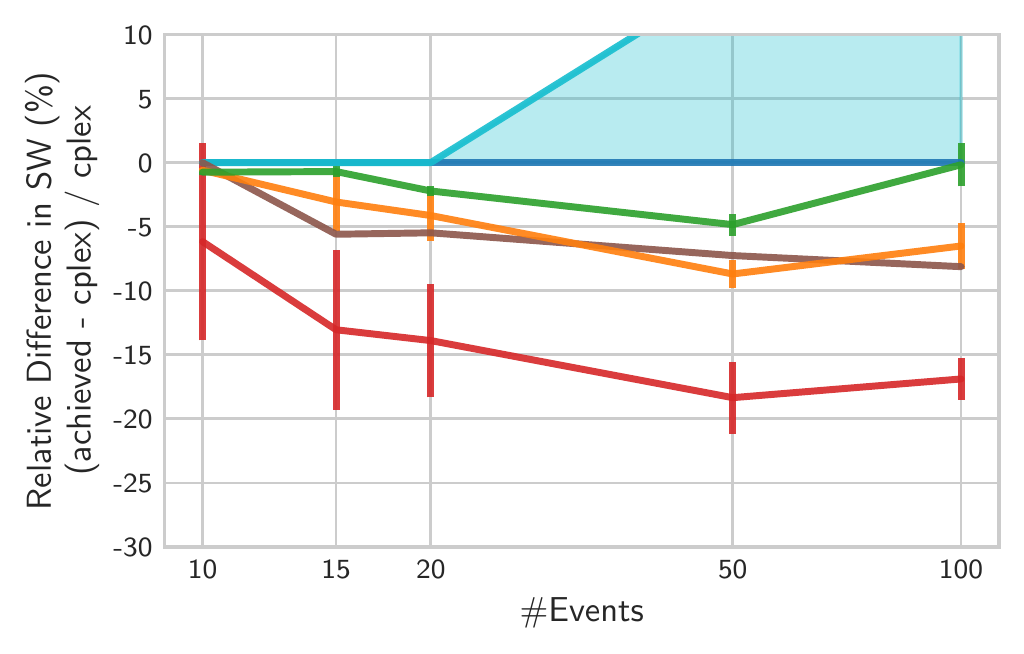}
		\subcaption{$p=20$}
	\end{subfigure}
	~ 
	\begin{subfigure}[t]{0.32\textwidth}
		\centering
		\includegraphics[width = 1 \linewidth, trim={0em 0em 0em 0em}, clip]{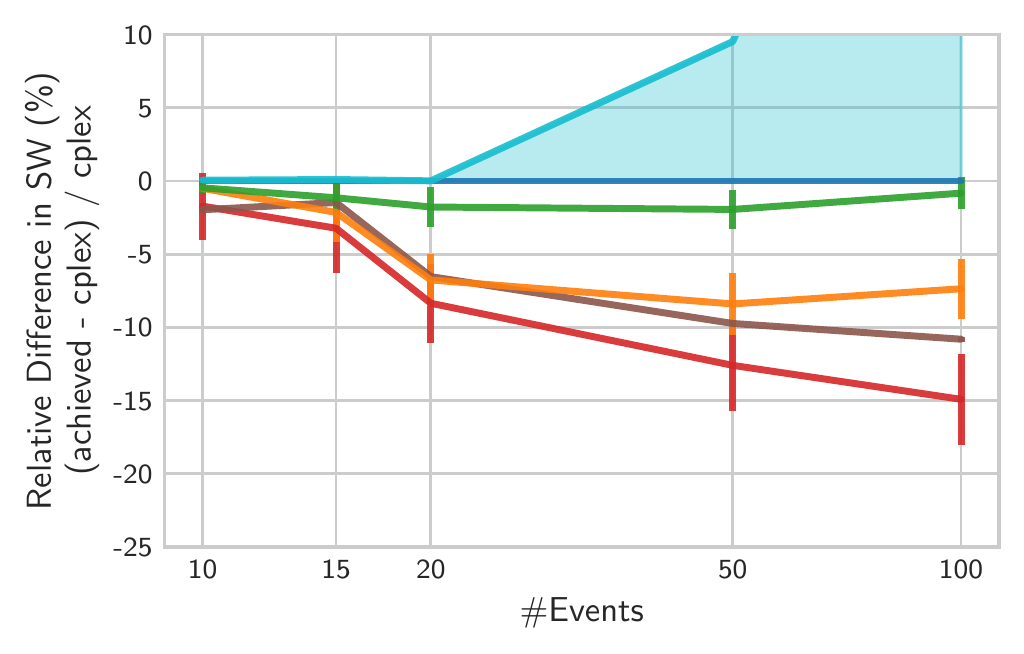}
		\subcaption{$p=30$}
	\end{subfigure}
	~ 
	\begin{subfigure}[t]{0.32\textwidth}
		\centering
		\includegraphics[width = 1 \linewidth, trim={0em 0em 0em 0em}, clip]{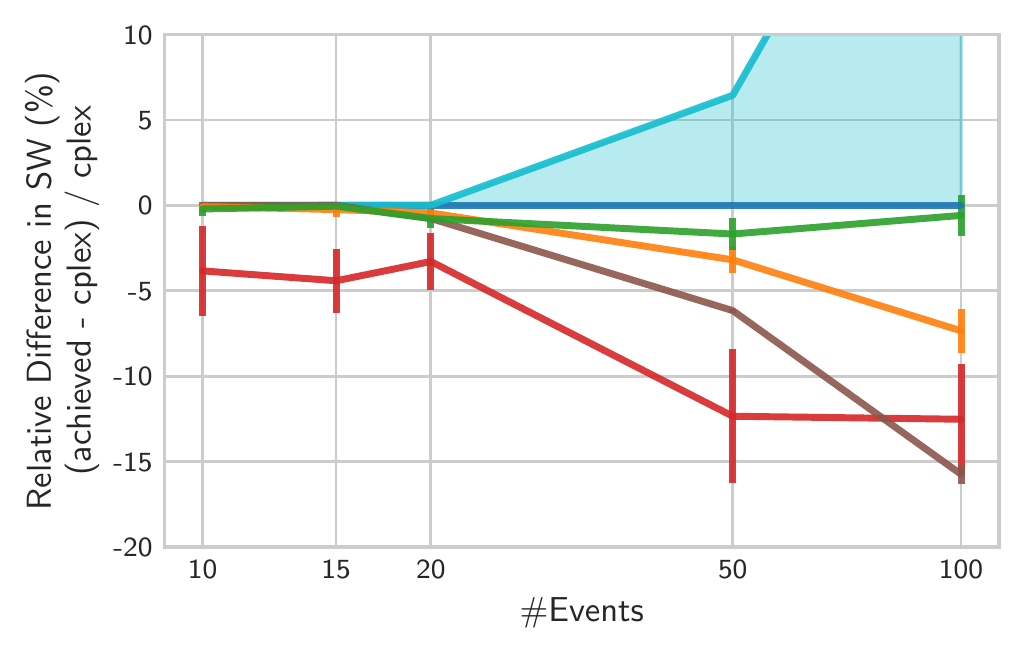}
		\subcaption{$p=50$}
	\end{subfigure}
	~ 
	\begin{subfigure}[t]{0.32\textwidth}
		\centering
		\includegraphics[width = 1 \linewidth, trim={0em 0em 0em 0em}, clip]{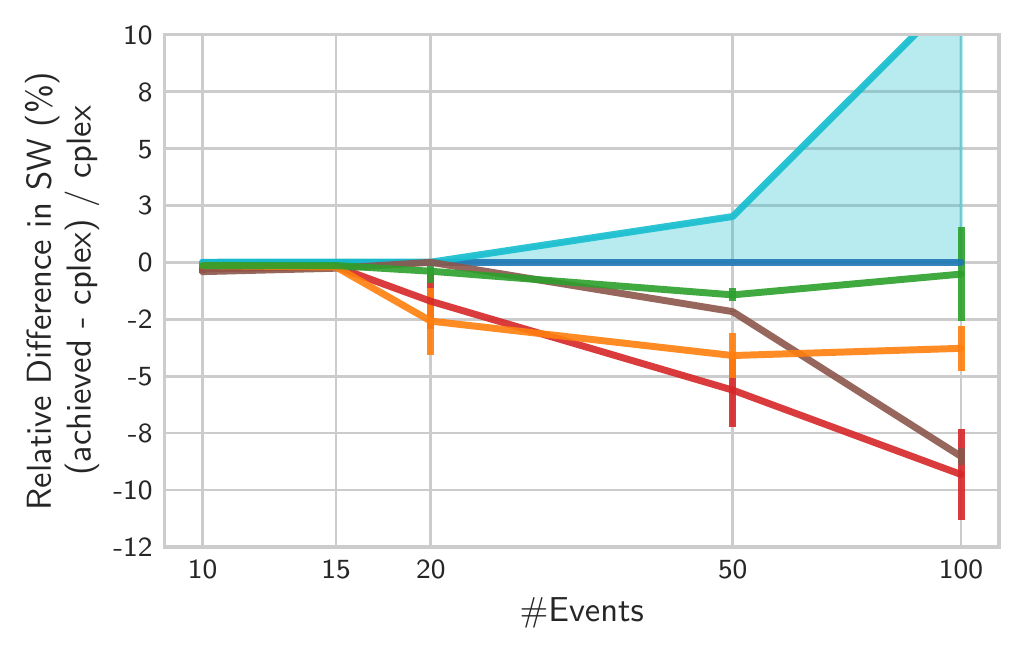}
		\subcaption{$p=100$}
	\end{subfigure}
	\begin{subfigure}[t]{0.32\textwidth}
		\vspace{-8 em}
		\includegraphics[scale = 0.8, trim={0em 0em 0em 0em}, clip]{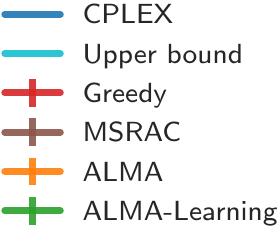}
	\end{subfigure}%
	\caption{Meeting Scheduling. Relative difference in social welfare (compared to CPLEX), for increasing number of events ($x$-axis in log scale). Results for various number of participants ($\mathcal{P}$). ALMA-Learning was trained for 512 time-steps.}
	\label{fig: appendixmeeting sw}
\end{figure*}

\begin{figure*}[t!]
	\centering
	\begin{subfigure}[t]{0.32\textwidth}
		\centering
		\includegraphics[width = 1 \linewidth, trim={0em 0em 0em 0em}, clip]{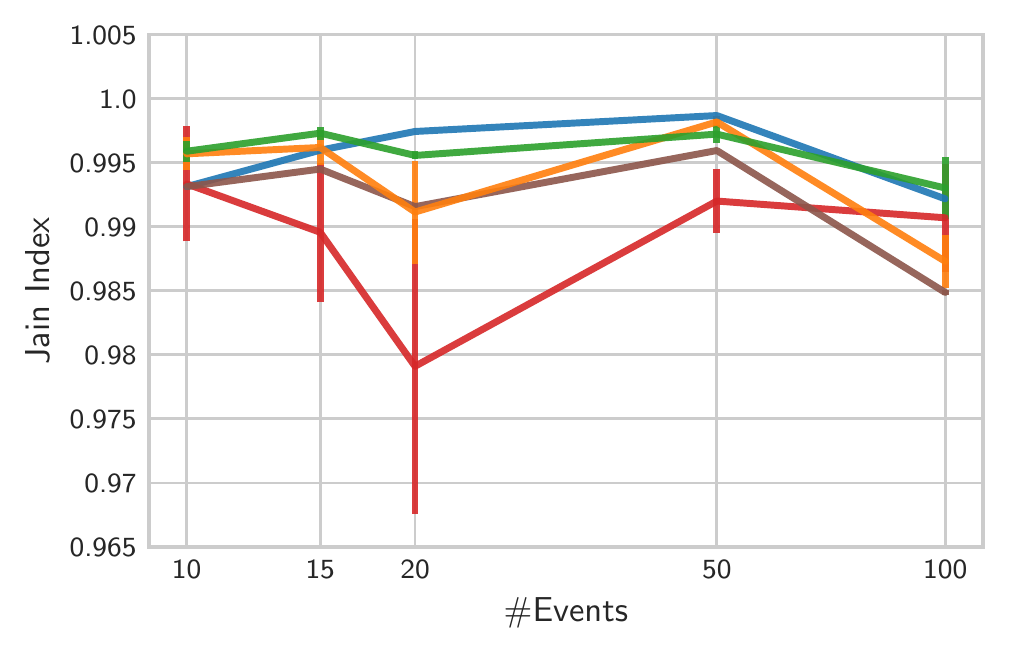}
		\subcaption{$p=20$}
	\end{subfigure}
	~ 
	\begin{subfigure}[t]{0.32\textwidth}
		\centering
		\includegraphics[width = 1 \linewidth, trim={0em 0em 0em 0em}, clip]{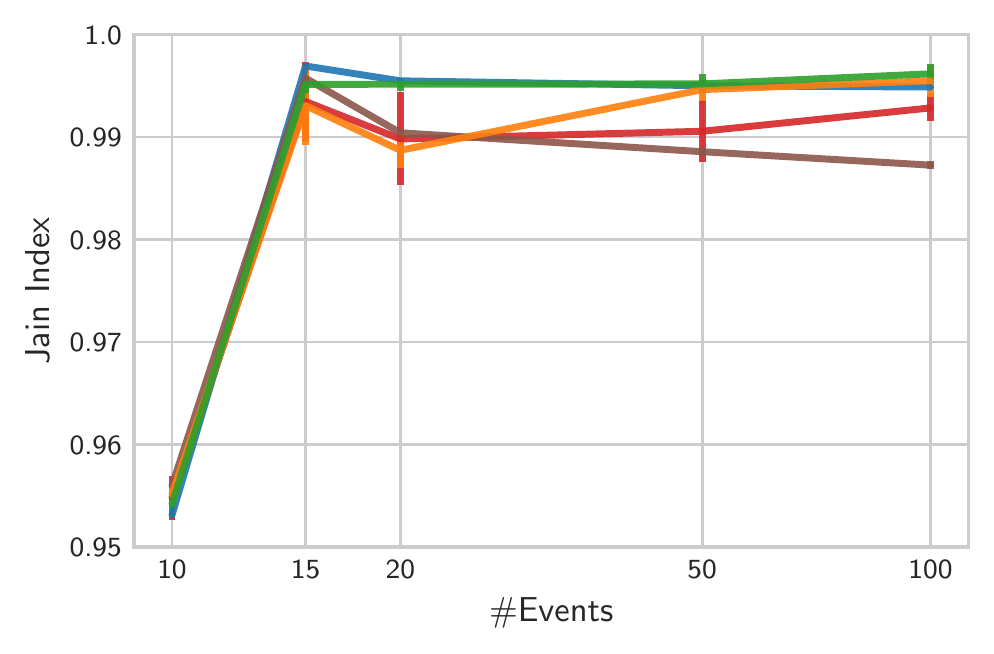}
		\subcaption{$p=30$}
	\end{subfigure}
	~ 
	\begin{subfigure}[t]{0.32\textwidth}
		\centering
		\includegraphics[width = 1 \linewidth, trim={0em 0em 0em 0em}, clip]{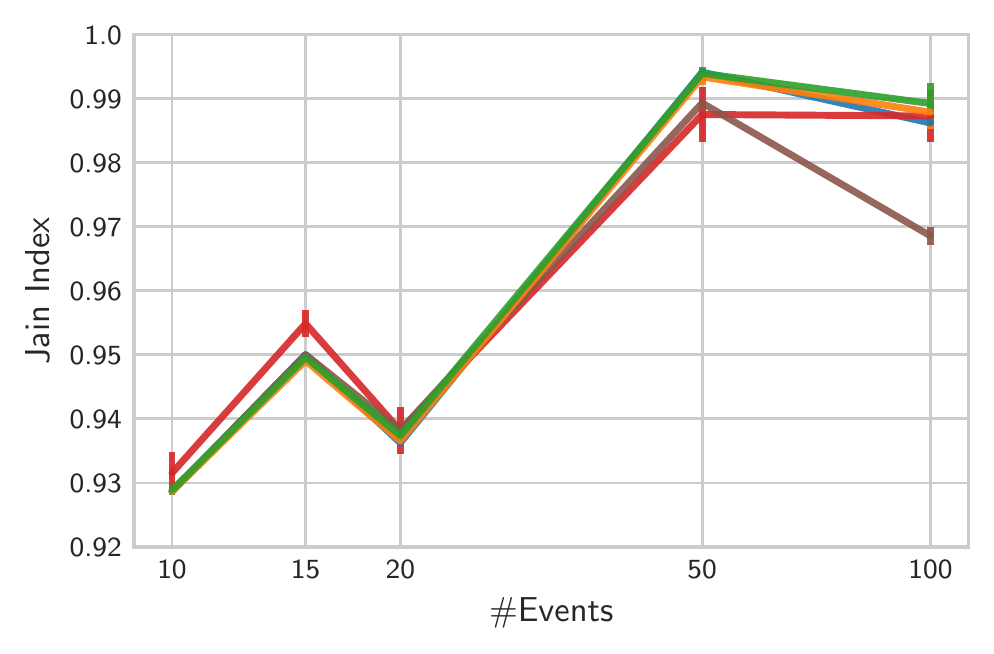}
		\subcaption{$p=50$}
	\end{subfigure}
	~ 
	\begin{subfigure}[t]{0.32\textwidth}
		\centering
		\includegraphics[width = 1 \linewidth, trim={0em 0em 0em 0em}, clip]{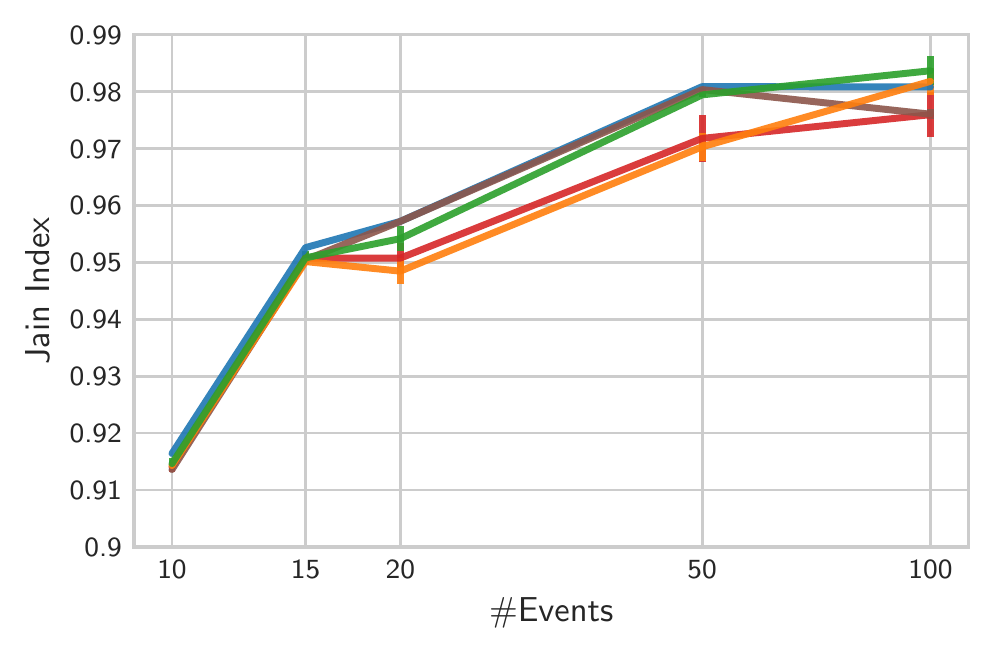}
		\subcaption{$p=100$}
	\end{subfigure}
	\begin{subfigure}[t]{0.32\textwidth}
		\vspace{-8 em}
		\includegraphics[scale = 0.8, trim={0em 0em 0em 0em}, clip]{./legend.pdf}
	\end{subfigure}%
	\caption{Meeting Scheduling. Fairness -- Jain index (the higher, the better), for increasing number of events, $\mathcal{E}$ ($x$-axis in log scale). Results for various number of participants ($\mathcal{P}$). ALMA-Learning was trained for 512 time-steps.}
	\label{fig: appendix meeting fairness jain}
\end{figure*}

\begin{figure*}[t!]
	\centering
	\begin{subfigure}[t]{0.32\textwidth}
		\centering
		\includegraphics[width = 1 \linewidth, trim={0em 0em 0em 0em}, clip]{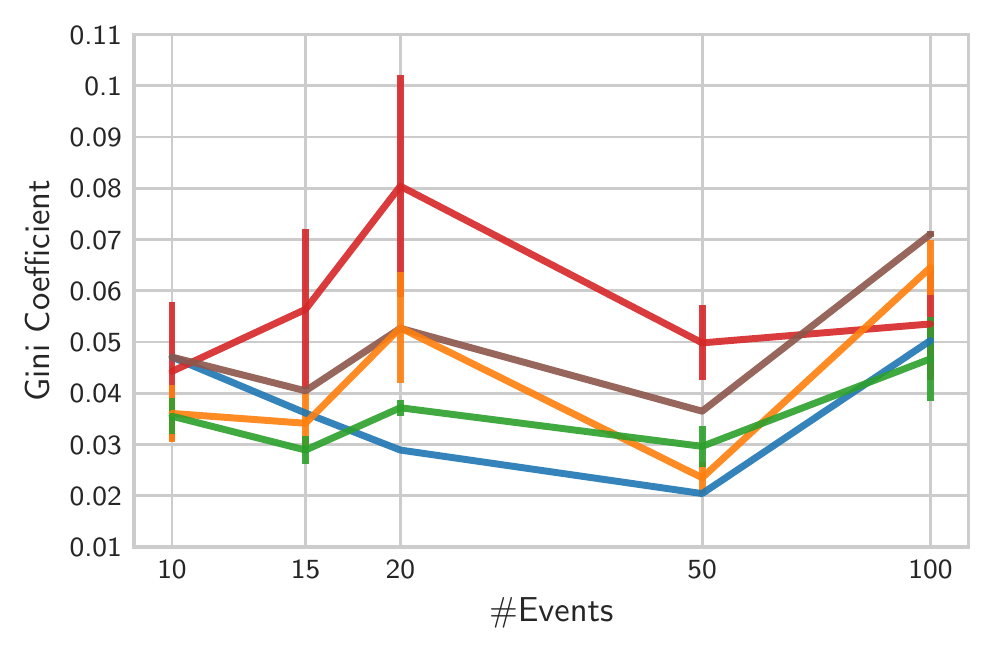}
		\subcaption{$p=20$}
	\end{subfigure}
	~ 
	\begin{subfigure}[t]{0.32\textwidth}
		\centering
		\includegraphics[width = 1 \linewidth, trim={0em 0em 0em 0em}, clip]{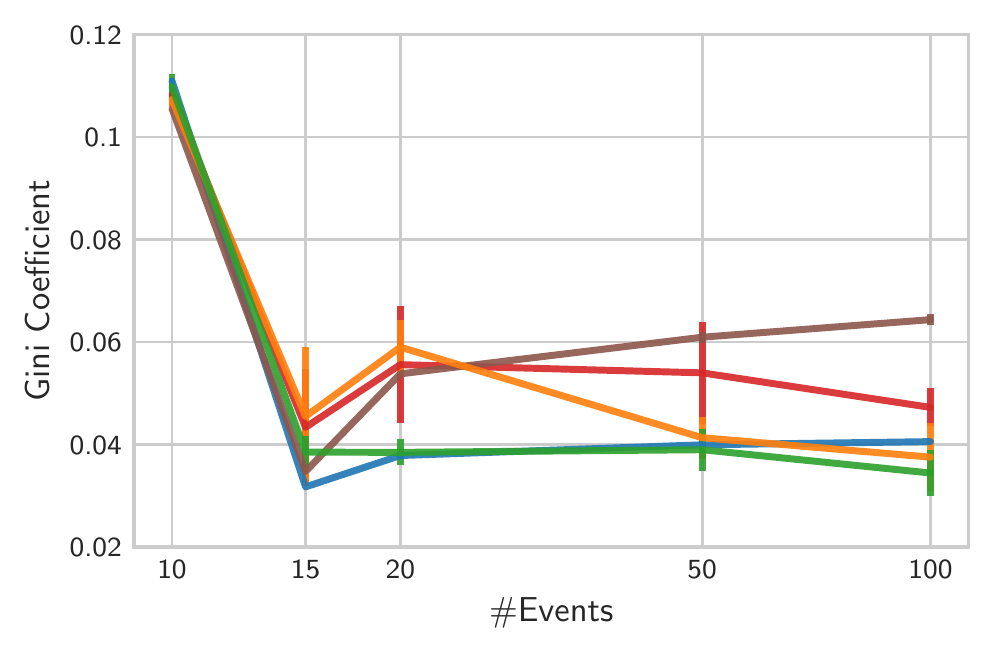}
		\subcaption{$p=30$}
	\end{subfigure}
	~ 
	\begin{subfigure}[t]{0.32\textwidth}
		\centering
		\includegraphics[width = 1 \linewidth, trim={0em 0em 0em 0em}, clip]{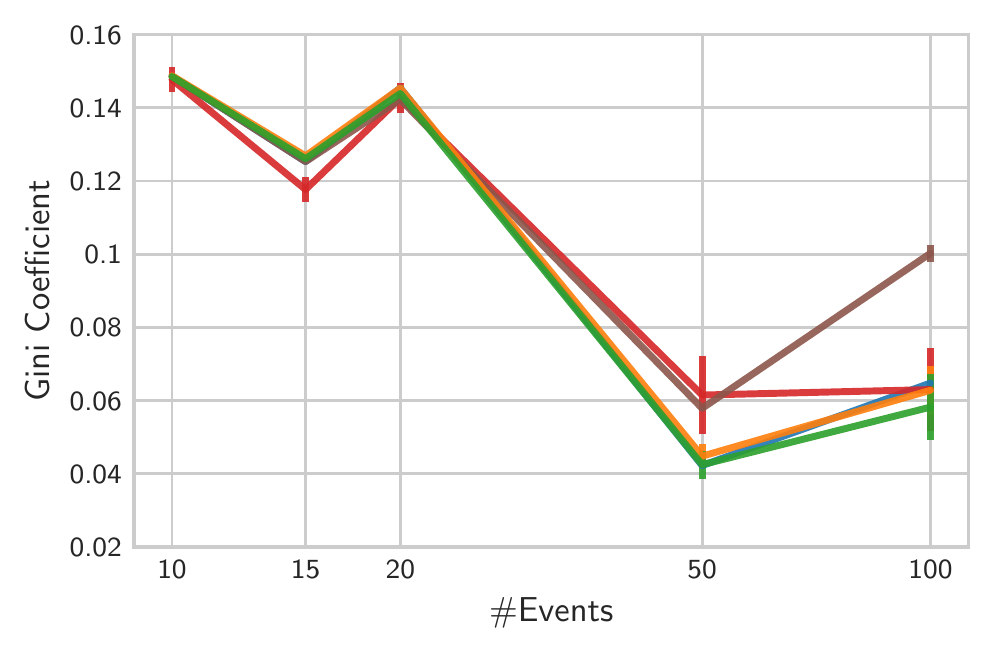}
		\subcaption{$p=50$}
	\end{subfigure}
	~ 
	\begin{subfigure}[t]{0.32\textwidth}
		\centering
		\includegraphics[width = 1 \linewidth, trim={0em 0em 0em 0em}, clip]{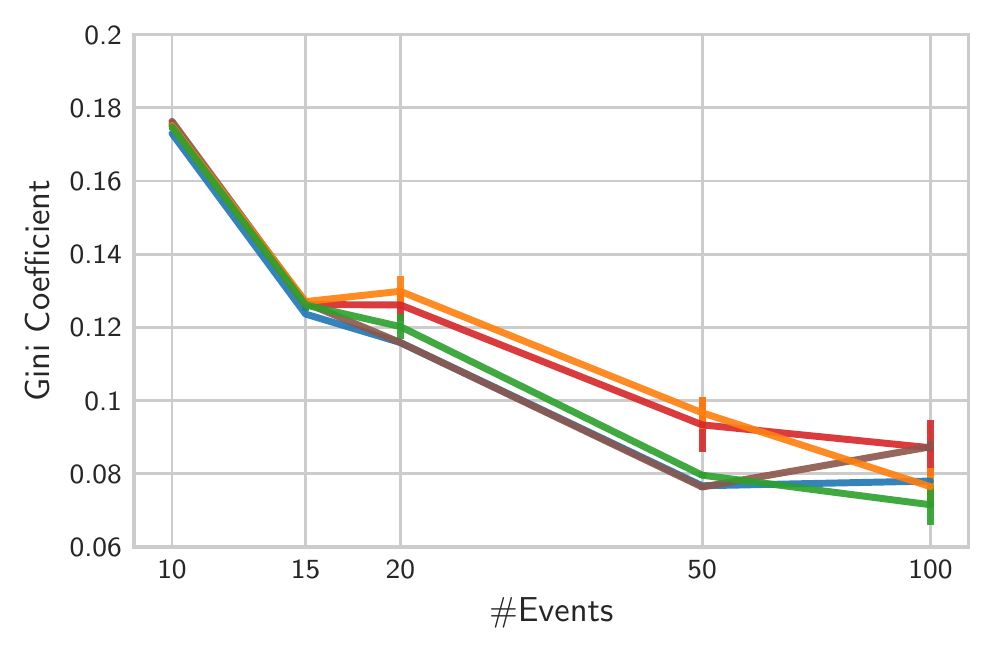}
		\subcaption{$p=100$}
	\end{subfigure}
	\begin{subfigure}[t]{0.32\textwidth}
		\vspace{-8 em}
		\includegraphics[scale = 0.8, trim={0em 0em 0em 0em}, clip]{./legend.pdf}
	\end{subfigure}%
	\caption{Meeting Scheduling. Fairness -- Gini coefficient (the lower, the better), for increasing number of events, $\mathcal{E}$ ($x$-axis in log scale). Results for various number of participants ($\mathcal{P}$). ALMA-Learning was trained for 512 time-steps.}
	\label{fig: appendix meeting fairness gini}
\end{figure*}

\begin{figure*}[t!]
	\centering
	\begin{subfigure}[t]{0.30\textwidth}
		\centering
		\includegraphics[width = 1 \linewidth, trim={0em 0em 0em 0em}, clip]{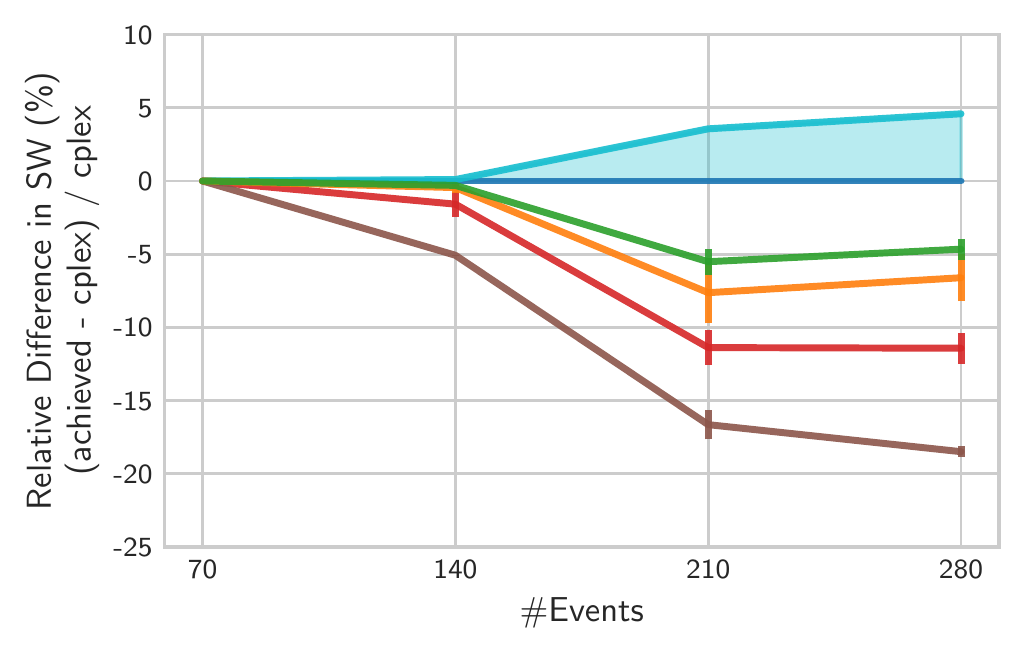}
		\subcaption{Relative difference in SW}
	\end{subfigure}
	~
	\begin{subfigure}[t]{0.30\textwidth}
		\centering
		\includegraphics[width = 1 \linewidth, trim={0em 0em 0em 0em}, clip]{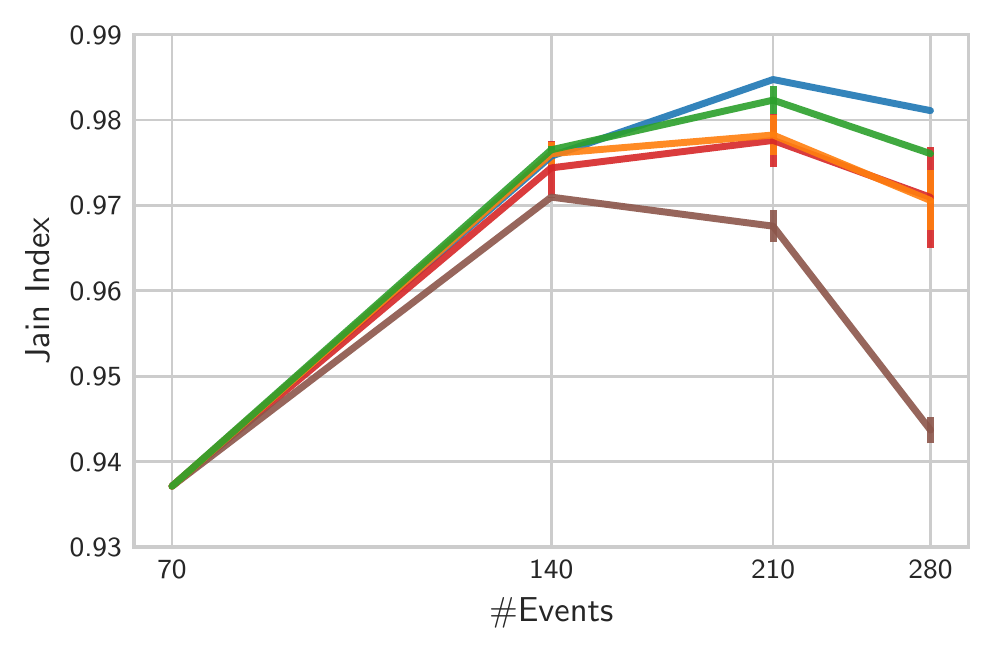}
		\subcaption{Jain index (higher is better)}
	\end{subfigure}
	~ 
	\begin{subfigure}[t]{0.30\textwidth}
		\centering
		\includegraphics[width = 1 \linewidth, trim={0em 0em 0em 0em}, clip]{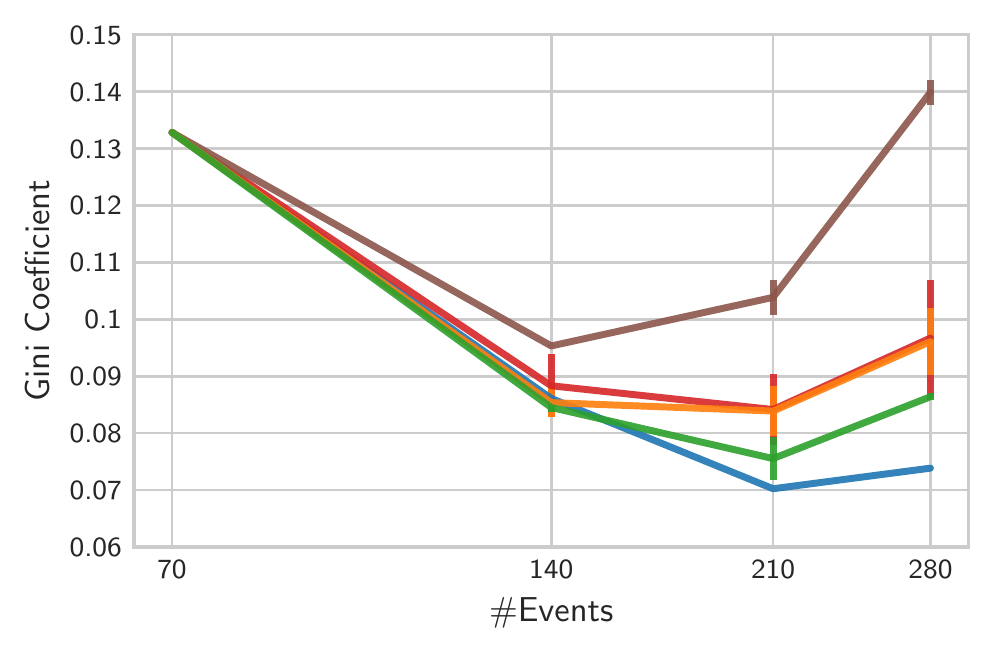}
		\subcaption{Gini coefficient (lower is better)}
	\end{subfigure}
	~
	\begin{subfigure}[t]{0.05\textwidth}
		\vspace{-7 em}
		\includegraphics[scale = 0.8, trim={0em 0em 0em 0em}, clip]{./legend.pdf}
	\end{subfigure}%
	\caption{Meeting Scheduling. Hand-crafted, large instances. Results for increasing number of events, $\mathcal{E}$ ($x$-axis in log scale), and 100 participants ($|\mathcal{P}| = 100$). ALMA-Learning was trained for 512 time-steps.}
	\label{fig: appendix meeting large}
\end{figure*}

\begin{table*}[t!]
	\centering
	\caption{Range of the average loss ($\%$) in social welfare compared to the CPLEX for increasing number of participants, $\mathcal{P}$ ($|\mathcal{E}| \in [10, 100]$). The last line corresponds to the loss compared to the upper bound for the optimal solution for the hand-crafted large test-case with $|\mathcal{P}| = 100, |\mathcal{E}| = 280$.}
	\label{tb: supp meeting scheduling social welfare}
	\resizebox{\textwidth}{!}{%
	\begin{tabular}{@{}rccccc@{}}
	\toprule
	                      & CPLEX               & Greedy                                 & MSRAC                                  & ALMA                                   & \textbf{ALMA-Learning}                \\ \midrule
	$|\mathcal{P}| = 20$  & N/A                 & $6.16 \pm 7.71 \% - 18.35 \pm 2.80 \%$ & $0.00 \pm 0.00 \% - 8.12 \pm 0.06 \%$  & $0.59 \pm 0.15 \% - 8.69 \pm 1.10 \%$  & $0.16 \pm 1.68 \% - 4.84 \pm 0.85 \%$ \\
	$|\mathcal{P}| = 50$  & N/A                 & $1.72 \pm 2.29 \% - 14.92 \pm 3.09 \%$ & $1.47 \pm 0.00 \% - 10.81 \pm 0.16 \%$ & $0.50 \pm 0.24 \% - 8.40 \pm 2.09 \%$  & $0.47 \pm 0.37 \% - 1.94 \pm 1.34 \%$ \\
	$|\mathcal{P}| = 50$  & N/A                 & $0.00 \pm 0.00 \% - 15.74 \pm 0.55 \%$ & $3.29 \pm 1.69 \% - 12.52 \pm 3.25 \%$ & $0.07 \pm 0.07 \% - 7.34 \pm 1.27 \%$  & $0.05 \pm 0.05 \% - 1.68 \pm 0.93 \%$ \\
	$|\mathcal{P}| = 100$ & N/A                 & $0.19 \pm 0.11 \% - 9.32 \pm 1.98 \%$  & $0.00 \pm 0.00 \% - 8.52 \pm 0.37 \%$  & $0.14 \pm 0.00 \% - 4.09 \pm 0.99 \%$  & $0.14 \pm 0.11 \% - 1.43 \pm 0.28 \%$ \\
	$|\mathcal{E}| = 280$ & $0.00 \% - 4.39 \%$ & $0.00 \pm 0.00 \% - 15.31 \pm 1.00 \%$ & $0.00 \pm 0.00 \% - 22.07 \pm 0.34 \%$ & $0.00 \pm 0.00 \% - 10.81 \pm 2.02 \%$ & $0.00 \pm 0.00 \% - 8.84 \pm 0.68 \%$ \\ \bottomrule
	\end{tabular}%
	}
\end{table*}

\begin{landscape}
\begin{table}[t!]
	\centering
	\caption{Meeting Scheduling. Fairness -- Jain index (the higher, the better), for increasing number of participants, $\mathcal{P}$ ($|\mathcal{E}| \in [10, 100]$). The last line corresponds to the hand-crafted large test-case with $|\mathcal{P}| = 100, |\mathcal{E}| = 280$. In parenthesis we include the average improvement in fairness compared to CPLEX.}
	\label{tb: supp meeting scheduling jain}
	\resizebox{0.8\paperheight}{!}{%
	\begin{tabular}{@{}rccccc@{}}
	\toprule
	                      & CPLEX         & Greedy                                                                   & MSRAC                                                                    & ALMA                                                                      & \textbf{ALMA-Learning}                                                   \\ \midrule
	$|\mathcal{P}| = 20$  & $0.99 - 1.00$ & $0.98 \pm 0.01 - 0.99 \pm 0.00 (98.16 \pm 1.16 \% - 100.02 \pm 0.45 \%)$ & $0.98 \pm 0.00 - 1.00 \pm 0.00 (99.26 \pm 0.02 \% - 100.00 \pm 0.00 \%)$ & $0.99 \pm 0.00 - 1.00 \pm 0.00 (99.37 \pm 99.37 \% - 100.26 \pm 0.13 \%)$ & $0.99 \pm 0.00 - 1.00 \pm 0.00 (99.81 \pm 0.03 \% - 100.28 \pm 0.08 \%)$ \\
	$|\mathcal{P}| = 50$  & $0.95 - 1.00$ & $0.95 \pm 0.00 - 0.99 \pm 0.00 (99.43 \pm 0.46 \% - 100.17 \pm 0.23 \%)$ & $0.96 \pm 0.00 - 1.00 \pm 0.00 (99.23 \pm 0.04 \% - 100.30 \pm 0.00 \%)$ & $0.96 \pm 0.00 - 1.00 \pm 0.00 (99.32 \pm 0.17 \% - 100.21 \pm 0.06 \%)$  & $0.95 \pm 0.00 - 1.00 \pm 0.00 (99.82 \pm 0.08 \% - 100.13 \pm 0.09 \%)$ \\
	$|\mathcal{P}| = 50$  & $0.93 - 0.99$ & $0.93 \pm 0.00 - 0.99 \pm 0.00 (99.34 \pm 0.43 \% - 100.51 \pm 0.23 \%)$ & $0.93 \pm 0.00 - 0.99 \pm 0.00 (98.21 \pm 0.14 \% - 100.21 \pm 0.00 \%)$ & $0.93 \pm 0.00 - 0.99 \pm 0.00 (99.90 \pm 99.90 \% - 100.17 \pm 0.26 \%)$ & $0.93 \pm 0.00 - 0.99 \pm 0.00 (99.95 \pm 0.03 \% - 100.31 \pm 0.32 \%)$ \\
	$|\mathcal{P}| = 100$ & $0.92 - 0.98$ & $0.91 \pm 0.00 - 0.98 \pm 0.00 (99.07 \pm 0.42 \% - 99.81 \pm 0.10 \%)$  & $0.91 \pm 0.00 - 0.98 \pm 0.00 (99.51 \pm 0.08 \% - 100.00 \pm 0.00 \%)$ & $0.91 \pm 0.00 - 0.98 \pm 0.00 (98.92 \pm 98.92 \% - 100.10 \pm 0.24 \%)$ & $0.91 \pm 0.00 - 0.98 \pm 0.00 (99.68 \pm 0.23 \% - 100.29 \pm 0.26 \%)$ \\
	$|\mathcal{E}| = 280$ & $0.94 - 0.98$ & $0.94 \pm 0.00 - 0.98 \pm 0.00 (98.97 \pm 0.60 \% - 100.00 \pm 0.00 \%)$ & $0.94 \pm 0.00 - 0.97 \pm 0.00 96.19 \pm 0.16 \% - 100.00 \pm 0.00 \%)$  & $0.94 \pm 0.00 - 0.98 \pm 0.00 (98.93 \pm 98.93 \% - 100.03 \pm 0.14 \%)$ & $0.94 \pm 0.00 - 0.98 \pm 0.00 (99.49 \pm 0.01 \% - 100.08 \pm 0.04 \%)$ \\ \bottomrule
	\end{tabular}%
	}
\end{table}

\begin{table}[t!]
	\centering
	\caption{Meeting Scheduling. Fairness -- Gini coefficient (the lower, the better), for increasing number of participants, $\mathcal{P}$ ($|\mathcal{E}| \in [10, 100]$). The last line corresponds to the hand-crafted large test-case with $|\mathcal{P}| = 100, |\mathcal{E}| = 280$. In parenthesis we include the average improvement in fairness compared to CPLEX.}
	\label{tb: supp meeting scheduling gini}
	\resizebox{0.8\paperheight}{!}{%
	\begin{tabular}{@{}rccccc@{}}
	\toprule
	                      & CPLEX         & Greedy                                                                     & MSRAC                                                                     & ALMA                                                                       & \textbf{ALMA-Learning}                                                    \\ \midrule
	$|\mathcal{P}| = 20$  & $0.02 - 0.05$ & $0.04 \pm 0.01 - 0.08 \pm 0.02 (94.07 \pm 28.94 \% - 278.19 \pm 75.03 \%)$ & $0.04 \pm 0.00 - 0.07 \pm 0.00 (100.00 \pm 0.00 \% - 182.22 \pm 0.00 \%)$ & $0.02 \pm 0.00 - 0.06 \pm 0.01 (76.63 \pm 76.63 \% - 182.55 \pm 37.40 \%)$ & $0.03 \pm 0.00 - 0.05 \pm 0.01 (75.48 \pm 7.50 \% - 145.14 \pm 19.93 \%)$ \\
	$|\mathcal{P}| = 50$  & $0.03 - 0.11$ & $0.04 \pm 0.01 - 0.11 \pm 0.00 (97.62 \pm 2.09 \% - 146.89 \pm 30.29 \%)$  & $0.03 \pm 0.00 - 0.11 \pm 0.00 (95.15 \pm 0.00 \% - 158.71 \pm 2.79 \%)$ & $0.04 \pm 0.01 - 0.11 \pm 0.00 (92.52 \pm 92.52 \% - 155.91 \pm 14.12 \%)$ & $0.03 \pm 0.00 - 0.11 \pm 0.00 (84.89 \pm 11.24 \% - 121.39 \pm 9.84 \%)$ \\
	$|\mathcal{P}| = 50$  & $0.04 - 0.15$ & $0.06 \pm 0.01 - 0.15 \pm 0.00 (93.90 \pm 2.76 \% - 145.61 \pm 25.23 \%)$  & $0.06 \pm 0.00 - 0.15 \pm 0.00 (97.86 \pm 0.00 \% - 154.72 \pm 3.55 \%)$  & $0.04 \pm 0.00 - 0.15 \pm 0.00 (96.95 \pm 96.95 \% - 106.06 \pm 7.60 \%)$  & $0.04 \pm 0.00 - 0.15 \pm 0.00 (89.84 \pm 13.96 \% - 100.62 \pm 0.45 \%)$ \\
	$|\mathcal{P}| = 100$ & $0.08 - 0.17$ & $0.09 \pm 0.01 - 0.18 \pm 0.00 (101.53 \pm 0.62 \% - 121.70 \pm 9.79 \%)$  & $0.08 \pm 0.00 - 0.18 \pm 0.00 (99.56 \pm 0.99 \% - 112.00 \pm 1.91 \%)$  & $0.08 \pm 0.00 - 0.18 \pm 0.00 (98.04 \pm 98.04 \% - 125.99 \pm 5.74 \%)$  & $0.07 \pm 0.01 - 0.17 \pm 0.00 (91.70 \pm 6.95 \% - 103.77 \pm 0.98 \%)$  \\
	$|\mathcal{E}| = 280$ & $0.07 - 0.13$ & $0.08 \pm 0.01 - 0.13 \pm 0.00 (100.00 \pm 0.00 \% - 130.96 \pm 13.86 \%)$ & $0.10 \pm 0.00 - 0.14 \pm 0.00 (100.00 \pm 0.00 \% - 189.42 \pm 3.03 \%)$ & $0.08 \pm 0.00 - 0.13 \pm 0.00 (99.14 \pm 99.14 \% - 130.12 \pm 8.07 \%)$  & $0.08 \pm 0.00 - 0.13 \pm 0.00 (98.15 \pm 0.93 \% - 116.99 \pm 0.83 \%)$  \\ \bottomrule
	\end{tabular}%
	}
\end{table}
\end{landscape}

\clearpage
\bibliographystyle{named}
\bibliography{arXiv_alma_learning_bibliography}

\begin{thebibliography}{}

\bibitem[\protect\citeauthoryear{Auer \bgroup \em et al.\egroup
  }{2002}]{auer2002nonstochastic}
Peter Auer, Nicolo Cesa-Bianchi, Yoav Freund, and Robert~E Schapire.
\newblock The nonstochastic multiarmed bandit problem.
\newblock {\em SIAM journal on computing}, 32(1):48--77, 2002.

\bibitem[\protect\citeauthoryear{Bellman}{2013}]{bellman2013dynamic}
Richard Bellman.
\newblock {\em Dynamic programming}.
\newblock Courier Corporation, 2013.

\bibitem[\protect\citeauthoryear{BenHassine and Ho}{2007}]{benhassine07}
Ahlem BenHassine and Tu~Bao Ho.
\newblock An agent-based approach to solve dynamic meeting scheduling problems
  with preferences.
\newblock {\em Engineering Applications of Artificial Intelligence},
  20(6):857--873, 2007.

\bibitem[\protect\citeauthoryear{B{\"u}rger \bgroup \em et al.\egroup
  }{2012}]{burger2012distributed}
Mathias B{\"u}rger, Giuseppe Notarstefano, Francesco Bullo, and Frank
  Allg{\"o}wer.
\newblock A distributed simplex algorithm for degenerate linear programs and
  multi-agent assignments.
\newblock {\em Automatica}, 2012.

\bibitem[\protect\citeauthoryear{Busoniu \bgroup \em et al.\egroup
  }{2008}]{Busoniu2008CSM22204332221106}
L.~Busoniu, R.~Babuska, and B.~De~Schutter.
\newblock A comprehensive survey of multiagent reinforcement learning.
\newblock {\em Trans. Sys. Man Cyber Part C}, 38(2):156--172, March 2008.

\bibitem[\protect\citeauthoryear{Crawford and
  Veloso}{2005}]{crawford2005learning}
Elisabeth Crawford and Manuela Veloso.
\newblock Learning dynamic preferences in multi-agent meeting scheduling.
\newblock In {\em IEEE/WIC/ACM International Conference on Intelligent Agent
  Technology}, pages 487--490. IEEE, 2005.

\bibitem[\protect\citeauthoryear{Danassis \bgroup \em et al.\egroup
  }{2019a}]{ijcai201931}
Panayiotis Danassis, Aris Filos-Ratsikas, and Boi Faltings.
\newblock Anytime heuristic for weighted matching through altruism-inspired
  behavior.
\newblock In {\em Proceedings of the Twenty-Eighth International Joint
  Conference on Artificial Intelligence, {IJCAI-19}}, pages 215--222, 2019.

\bibitem[\protect\citeauthoryear{Danassis \bgroup \em et al.\egroup
  }{2019b}]{danassis2019putting}
Panayiotis Danassis, Marija Sakota, Aris Filos-Ratsikas, and Boi Faltings.
\newblock Putting ridesharing to the test: Efficient and scalable solutions and
  the power of dynamic vehicle relocation.
\newblock {\em ArXiv: 1912.08066}, 2019.

\bibitem[\protect\citeauthoryear{Danassis \bgroup \em et al.\egroup
  }{2020}]{danassis2020differential}
Panayiotis Danassis, Aleksei Triastcyn, and Boi Faltings.
\newblock Differential privacy meets maximum-weight matching, 2020.

\bibitem[\protect\citeauthoryear{Elkin}{2004}]{Elkin2004DAS10549161054931}
Michael Elkin.
\newblock Distributed approximation: A survey.
\newblock {\em SIGACT News}, 35(4):40--57, December 2004.

\bibitem[\protect\citeauthoryear{Franzin \bgroup \em et al.\egroup
  }{2002}]{franzin02}
Maria~Sole Franzin, EC~Freuder, F~Rossi, and R~Wallace.
\newblock Multi-agent meeting scheduling with preferences: efficiency, privacy
  loss, and solution quality.
\newblock In {\em Proceedings of the AAAI Workshop on Preference in AI and CP},
  2002.

\bibitem[\protect\citeauthoryear{Geng and Cassandras}{2013}]{geng2013new}
Yanfeng Geng and Christos~G. Cassandras.
\newblock New “smart parking” system based on resource allocation and
  reservations.
\newblock {\em IEEE Transactions on Intelligent Transportation Systems}, 2013.

\bibitem[\protect\citeauthoryear{Gini}{1912}]{gini1912variabilita}
Corrado Gini.
\newblock Variabilit{\`a} e mutabilit{\`a}.
\newblock {\em Reprinted in Memorie di metodologica statistica (Ed. Pizetti E,
  Salvemini, T). Rome: Libreria Eredi Virgilio Veschi}, 1912.

\bibitem[\protect\citeauthoryear{Giordani \bgroup \em et al.\egroup
  }{2010}]{giordani2010distributed}
Stefano Giordani, Marin Lujak, and Francesco Martinelli.
\newblock A distributed algorithm for the multi-robot task allocation problem.
\newblock In {\em International Conference on Industrial, Engineering and Other
  Applications of Applied Intelligent Systems}. Springer, 2010.

\bibitem[\protect\citeauthoryear{Gunn and Anderson}{2013}]{GUNN201322}
Tyler Gunn and John Anderson.
\newblock Dynamic heterogeneous team formation for robotic urban search and
  rescue.
\newblock {\em Procedia Computer Science}, 2013.
\newblock The 4th Int. Conf. on Ambient Systems, Networks and Technologies (ANT
  2013), the 3rd Int. Conf. on Sustainable Energy Information Technology
  (SEIT-2013).

\bibitem[\protect\citeauthoryear{Hassine \bgroup \em et al.\egroup
  }{2004}]{10.5555/1018411.1018882}
Ahlem~Ben Hassine, Xavier Defago, and Tu~Bao Ho.
\newblock Agent-based approach to dynamic meeting scheduling problems.
\newblock In {\em Proceedings of the Third International Joint Conference on
  Autonomous Agents and Multiagent Systems - Volume 3}, AAMAS '04, page
  1132–1139, USA, 2004. IEEE Computer Society.

\bibitem[\protect\citeauthoryear{Hutter \bgroup \em et al.\egroup
  }{2011}]{hutter11}
Frank Hutter, Holger~H. Hoos, and Kevin Leyton-Brown.
\newblock Sequential model-based optimization for general algorithm
  configuration.
\newblock In Carlos A.~Coello Coello, editor, {\em Learning and Intelligent
  Optimization}, pages 507--523, Berlin, Heidelberg, 2011. Springer Berlin
  Heidelberg.

\bibitem[\protect\citeauthoryear{Ismail and Sun}{2017}]{7991447}
Sarah Ismail and Liang Sun.
\newblock Decentralized hungarian-based approach for fast and scalable task
  allocation.
\newblock In {\em 2017 Int. Conf. on Unmanned Aircraft Systems (ICUAS)}, 2017.

\bibitem[\protect\citeauthoryear{Jain \bgroup \em et al.\egroup
  }{1998}]{DBLP:journals/corr/cs-NI-9809099}
Raj Jain, Dah{-}Ming Chiu, and W.~Hawe.
\newblock A quantitative measure of fairness and discrimination for resource
  allocation in shared computer systems.
\newblock {\em CoRR}, cs.NI/9809099, 1998.

\bibitem[\protect\citeauthoryear{Kuhn \bgroup \em et al.\egroup
  }{2016}]{Kuhn2016LCL29061422742012}
Fabian Kuhn, Thomas Moscibroda, and Roger Wattenhofer.
\newblock Local computation: Lower and upper bounds.
\newblock {\em J. ACM}, 63(2):17:1--17:44, March 2016.

\bibitem[\protect\citeauthoryear{Kuhn}{1955}]{kuhn1955hungarian}
Harold~W. Kuhn.
\newblock The hungarian method for the assignment problem.
\newblock {\em Naval Research Logistics}, 1955.

\bibitem[\protect\citeauthoryear{Laborie \bgroup \em et al.\egroup
  }{2018}]{laborie2018ibm}
Philippe Laborie, J{\'e}r{\^o}me Rogerie, Paul Shaw, and Petr Vil{\'\i}m.
\newblock Ibm ilog cp optimizer for scheduling.
\newblock {\em Constraints}, 23(2):210--250, 2018.

\bibitem[\protect\citeauthoryear{Lov{\'a}sz and
  Plummer}{2009}]{lovasz2009matching}
L{\'a}szl{\'o} Lov{\'a}sz and Michael~D. Plummer.
\newblock {\em Matching theory}, volume 367.
\newblock American Mathematical Soc., 2009.

\bibitem[\protect\citeauthoryear{Maheswaran \bgroup \em et al.\egroup
  }{2004}]{maheswaran04}
Rajiv~T. Maheswaran, Milind Tambe, Emma Bowring, Jonathan~P. Pearce, and
  Pradeep Varakantham.
\newblock Taking dcop to the real world: Efficient complete solutions for
  distributed multi-event scheduling.
\newblock In {\em Proceedings of the Third International Joint Conference on
  Autonomous Agents and Multiagent Systems - Volume 1}, AAMAS '04, page
  310–317, USA, 2004. IEEE Computer Society.

\bibitem[\protect\citeauthoryear{Munkres}{1957}]{munkres1957algorithms}
James Munkres.
\newblock Algorithms for the assignment and transportation problems.
\newblock {\em Journal of the society for industrial and applied mathematics},
  1957.

\bibitem[\protect\citeauthoryear{Nigam and Srivastava}{2020}]{nigam20}
Archana Nigam and Sanjay Srivastava.
\newblock Oddms: Online distributed dynamic meeting scheduler.
\newblock In Simon Fong, Nilanjan Dey, and Amit Joshi, editors, {\em ICT
  Analysis and Applications}, pages 199--208, Singapore, 2020. Springer
  Singapore.

\bibitem[\protect\citeauthoryear{Ottens \bgroup \em et al.\egroup
  }{2017}]{ottens12}
Brammert Ottens, Christos Dimitrakakis, and Boi Faltings.
\newblock Duct: An upper confidence bound approach to distributed constraint
  optimization problems.
\newblock {\em ACM Trans. Intell. Syst. Technol.}, 8(5), July 2017.

\bibitem[\protect\citeauthoryear{Romano and Nunamaker}{2001}]{romano01}
Nicholas~C Romano and Jay~F Nunamaker.
\newblock Meeting analysis: Findings from research and practice.
\newblock In {\em Proceedings of the 34th annual Hawaii international
  conference on system sciences}, pages 13--pp. IEEE, 2001.

\bibitem[\protect\citeauthoryear{Stone \bgroup \em et al.\egroup
  }{2010}]{AAAI10-adhoc}
Peter Stone, Gal~A.\ Kaminka, Sarit Kraus, and Jeffrey~S.\ Rosenschein.
\newblock Ad hoc autonomous agent teams: Collaboration without
  pre-coordination.
\newblock In {\em AAAI}, 2010.

\bibitem[\protect\citeauthoryear{Su}{2015}]{su2015algorithms}
Hsin-Hao Su.
\newblock Algorithms for fundamental problems in computer networks.
\newblock 2015.

\bibitem[\protect\citeauthoryear{Varakantham \bgroup \em et al.\egroup
  }{2012}]{varakantham2012decision}
Pradeep Varakantham, Shih-Fen Cheng, Geoff Gordon, and Asrar Ahmed.
\newblock Decision support for agent populations in uncertain and congested
  environments.
\newblock 2012.

\bibitem[\protect\citeauthoryear{Zavlanos \bgroup \em et al.\egroup
  }{2008}]{zavlanos2008distributed}
Michael~M. Zavlanos, Leonid Spesivtsev, and George~J Pappas.
\newblock A distributed auction algorithm for the assignment problem.
\newblock In {\em Decision and Control, 2008.} IEEE, 2008.

\bibitem[\protect\citeauthoryear{Zunino and Campo}{2009}]{zunino09}
Alejandro Zunino and Marcelo Campo.
\newblock Chronos: A multi-agent system for distributed automatic meeting
  scheduling.
\newblock {\em Expert Systems with Applications}, 36(3):7011--7018, 2009.

\end{thebibliography}
\end{document}